\documentclass[12pt,onecolumn,draftclsnofoot]{IEEEtran}
\usepackage{graphicx}
\usepackage{epsfig}
\usepackage{subfigure}
\usepackage{amssymb}
\usepackage{amsbsy}
\usepackage{amsmath}
\usepackage{cite}
\usepackage{url}

\usepackage{xcolor}
\usepackage{color}
\usepackage{float}
\usepackage{times,amsmath,epsfig}
\usepackage{xspace,latexsym,syntonly}
\usepackage{amssymb}
\usepackage{amsfonts}
\usepackage{textcomp}
\usepackage{subfigure}
\usepackage{amsbsy}
\usepackage{cite}

\newcommand{\E}{\ensuremath{\hbox{\textbf{E}}}}

\newtheorem{theorem}{Theorem}
\newtheorem{lemma}{Lemma}
\newtheorem{remark}{Remark}
\newtheorem{definition}{Definition}

\newcommand{\beq}{\begin{equation}}
\newcommand{\eeq}{\end{equation}}
\newcommand{\bea}{\begin{array}}
\newcommand{\ena}{\end{array}}
\newcommand{\bds}{\begin {itemize}}
\newcommand{\eds}{\end {itemize}}
\newcommand{\bdf}{\begin{definition}}
\newcommand{\blm}{\begin{lemma}}
\newcommand{\edf}{\end{definition}}
\newcommand{\elm}{\end{lemma}}
\newcommand{\bthm}{\begin{theorem}}
\newcommand{\ethm}{\end{theorem}}
\newcommand{\bprp}{\begin{prop}}
\newcommand{\eprp}{\end{prop}}
\newcommand{\bcl}{\begin{claim}}
\newcommand{\ecl}{\end{claim}}
\newcommand{\bcr}{\begin{coro}}
\newcommand{\ecr}{\end{coro}}
\newcommand{\bquest}{\begin{question}}
\newcommand{\equest}{\end{question}}

\newcommand{\larrow}{{\larrow}}

\newcommand{\nin}{{\not \in}}

\def\urltilda{\kern -.15em\lower .7ex\hbox{\~{}}\kern .04em}

\begin{document}
\title{Searching for Anomalies over \\ Composite Hypotheses}
\author{Bar Hemo, Tomer Gafni, Kobi Cohen, Qing Zhao
\thanks{Bar Hemo, Tomer Gafni and Kobi Cohen are with the School of Electrical and Computer Engineering, Ben-Gurion University of the Negev, Beer-Sheva 84105, Israel. Email: $\left\{\mbox{barh, gafnito, yakovsec}\right\}$@bgu.ac.il}
\thanks{Qing Zhao is with the School of Electrical and Computer Engineering, Cornell University, Ithaca, NY 14853, USA. Email: qz16@cornell.edu}
\thanks{Bar Hemo and Tomer Gafni contributed equally to this paper.}
\thanks{Part of this work was presented at the 16th IEEE International Symposium on Signal Processing and Information Technology (ISSPIT), Limassol, Cyprus, Dec. 2016 \cite{hemo2016asymptotically}.}
\thanks{This work has been submitted to the IEEE for possible publication. Copyright may be transferred without notice, after which this version may no longer be accessible.}
\thanks{The work of B. Hemo, T. Gafni, and K. Cohen was supported by the Israel National Cyber Bureau under grant 076/16, and by the U.S.-Israel Binational Science Foundation (BSF) under grant 2017723. The work of Q. Zhao was supported by the U.S. Army Research Office under Grant W911NF-17-1-0464, and the National Science Foundation  under Grant CCF-1815559.}
}
\date{}
\maketitle

\begin{abstract}
\label{sec:abstract}
The problem of detecting anomalies in multiple processes is considered. We consider a composite hypothesis case, in which the measurements drawn when observing a process follow a common distribution with an unknown parameter (vector), whose value lies in normal or abnormal parameter spaces, depending on its state. The objective is a sequential search strategy that minimizes the expected detection time subject to an error probability constraint. We develop a deterministic search algorithm with the following desired properties. First, when no additional side information on the process states is known, the proposed algorithm is asymptotically optimal in terms of minimizing the detection delay as the error probability approaches zero. Second, when the parameter value under the null hypothesis is known and equal for all normal processes, the proposed algorithm is asymptotically optimal as well, with better detection time determined by the true null state. Third, when the parameter value under the null hypothesis is unknown, but is known to be equal for all normal processes, the proposed algorithm is consistent in terms of achieving error probability that decays to zero with the detection delay. Finally, an explicit upper bound on the error probability under the proposed algorithm is established for the finite sample regime. Extensive experiments on synthetic dataset and DARPA intrusion detection dataset are conducted, demonstrating strong performance of the proposed algorithm over existing methods.
\end{abstract}
%
\IEEEpeerreviewmaketitle


\section{Introduction}
\label{sec:introduction}

We consider the problem of searching for an anomalous process (or few abnormal processes) among $M$ processes. For convenience, we often refer to the processes as cells and the anomalous process as the target which can locate in any of the $M$ cells. The decision maker is allowed to search for the target over $K$ cells at a time ($1\leq K\leq M$). We consider the composite hypothesis case, where the observation distribution has an unknown parameter (vector). When taking observations from a certain cell, random continuous values are measured which are drawn from a common distribution $f$. The distribution $f$ has an unknown parameter, belonging to parameter spaces $\Theta^{(0)}$ or $\Theta^{(1)}$, depending on whether the target is absent or present, respectively. The objective is a sequential search strategy that minimizes the expected detection time subject to an error probability constraint.
The anomaly detection problem finds applications in intrusion detection in cyber systems for quickly locating infected nodes by detecting statistical anomalies, spectrum scanning in cognitive radio networks for quickly detecting idle channels, and event detection in sensor networks.

\subsection{Main Results}
\label{ssec:main_results}

Dynamic search algorithms can be broadly classified into two classes:
(i)~Algorithms that use open-loop selection rules, in which the decision of which cell to search is predetermined and independent of the sequence of observations. The stopping rule, that decides when to stop collecting observations from the current cell, and whether to switch to the next cell or stop the test, however, is dynamically updated based on past observations. In this class of algorithms, tractable optimal solutions have been obtained under various settings of observation distributions (see e.g., \cite{Lai_2011_Quickest, Malloy_2012_Quickest, Tajer_2013_Quick, cohen2014optimal}).
(ii)~Algorithms that use closed-loop selection rules, in which the decision of which cell to search is based on past observations. The focus is on addressing the full-blown dynamic problem by jointly optimizing both selection and stopping rules in decision making (see e.g., \cite{Zigangirov_1966_Problem, Castanon_1995_Optimal, cohen2015active, vaidhiyan2015learning, nitinawarat2017universal, huang2018active, gurevich2019sequential}). In this setting, however, tractable optimal solutions have been obtained only for very special cases of observation distributions (\hspace{-0.01cm}\cite{Zigangirov_1966_Problem, Castanon_1995_Optimal}). In this paper we focus on the latter setting.

Since observations are drawn in a one-at-a-time manner, we are facing a sequential detection problem over multiple composite hypotheses. Sequential detection problems involving multiple processes are partially-observed Markov decision processes (POMDP) \cite{Castanon_1995_Optimal} which have exponential complexity in general. As a result, computing optimal search policies is intractable (except for some special cases of observation distributions as in \cite{Zigangirov_1966_Problem, Castanon_1995_Optimal}). When dealing with composite hypotheses, computing optimal policies is intractable even for the single process case. For tractability, a commonly adopted performance measure is asymptotic optimality in terms of minimizing the detection time as the error probability approaches zero (see, for example, classic and recent results in \cite{Chernoff_1959_Sequential, Schwarz_1962_Asymptotic, Lai_1988_Nearly, Pavlov_1990_Sequential, Tartakovsky_2002_Efficient, Nitinawarat_2013_Controlled, Naghshvar_2013_Active, naghshvar2013sequentiality, nitinawarat2015controlled, cohen2015active, vaidhiyan2015learning, song2016sequential, song2017asymptotically, nitinawarat2017universal}). The focus of this paper is thus on asymptotically optimal strategies with low computational complexity. Our main contributions are three fold, as detailed below.\\
\paragraph{A general model for composite hypotheses}

Dynamic search problems have been investigated under various models of observation distributions in past and recent years. Closed-loop solutions have been obtained under known Wiener processes \cite{Zigangirov_1966_Problem}, known symmetric distributions \cite{Castanon_1995_Optimal}, known general distributions \cite{cohen2015active}, known Poisson point processes with unknown parameters \cite{vaidhiyan2015learning}, and unknown distributions in which the measurements can take values from a finite alphabet \cite{nitinawarat2017universal}. By contrast to these works, in this paper the dynamic search is conducted over a general known distribution model with unknown parameters that lie in disjoint normal and abnormal parameter sets, and the measurements can take continuous values. This distribution model finds applications in traffic analysis in computer networks \cite{Thatte_2011_Parametric} and spectrum scanning in cognitive radio networks \cite{font2010glrt} for instance. Handling this observation model in the dynamic search setting leads to fundamentally different algorithm design and analysis as compared to existing methods.\\
\paragraph{Algorithm development} In terms of algorithm development, the proposed algorithm is deterministic and has low-complexity implementations. Specifically, the proposed algorithm consists of exploration and exploitation phases. During exploration, the cells are probed in a round-robin manner for learning the unknown parameters. During exploitation, the most informative observations are collected based on the estimated parameters. We point out that our algorithm uses only bounded exploration time under the setting without side information (Section \ref{ssec:uncorrelated}) and when the null hypothesis is assumed known (Section \ref{ssec:known}), which is of particular significance. It is in sharp contrast with the logarithmic order of exploration time commonly seen in active search strategies (see, for example, \cite{cohen2015asymptotically, nitinawarat2017universal} or even linear order of exploration time in \cite{vaidhiyan2015learning}).\\
\paragraph{Performance analysis} In terms of theoretical performance analysis, we prove that the proposed algorithm achieves asymptotic optimality when no additional side information on the process states is known, and a single location is probed at a time (as widely assumed in dynamic search studies for purposes of analysis, e.g., \cite{Zigangirov_1966_Problem, Tognetti_1968_An, Kadane_1971_Optimal, Castanon_1995_Optimal, Zhai_2013_Dynamic, vaidhiyan2015learning, nitinawarat2017universal}). 
Furthermore, when the parameter value under the null hypothesis is known (i.e., as widely applied in anomaly detection cases, and also assumed in \cite{nitinawarat2017universal} for establishing asymptotic optimality), we establish asymptotic optimality as well, with better detection time determined by the true null state. We also consider the case where the parameter value under the null hypothesis is unknown, but is identical for all normal processes. In this case, the proposed algorithm is shown to be consistent in terms of achieving error probability that decays to zero with time. In addition to the asymptotic analysis, an explicit upper bound on the error probability is established under the finite sample regime. Extensive numerical experiments on synthetic dataset and DARPA intrusion detection dataset have been conducted to demonstrate the efficiency of the proposed algorithm.

\subsection{Related work}
\label{ssec:related}
\noindent
Optimal solutions for target search or target whereabout problems have been obtained under some special cases when a single location is probed at a time.
Modern application areas of search problems with limited sensing resources include narrowband spectrum scanning \cite{egan2017fast,Zhao_2010_Quickest}, event detection by a fusion center that communicates with sensors using narrowband transmission \cite{blum2008energy,cohen2011energy}, and sensor  visual search studied recently by neuroscientists \cite{vaidhiyan2015learning}.
Results under the sequential setting can be found in~\cite{Zigangirov_1966_Problem, Klimko_1975_Optimal, Dragalin_1996_Simple, Stone_1971_Optimal, cohen2014optimal, blum2008energy, cohen2011energy}. Specifically, optimal policies were derived in~\cite{Zigangirov_1966_Problem, Klimko_1975_Optimal, Dragalin_1996_Simple} for the problem of quickest search over Wiener processes. In~\cite{Stone_1971_Optimal, cohen2014optimal}, optimal search strategies were established under the constraint that switching to a new process is allowed only when the state of the currently probed process is declared. Optimal policies under general distributions and unconstrained search model remain an open question. In this paper we address this question under the asymptotic regime as the error probability approaches zero. Optimal search strategies when a single location is probed at a time and a fixed sample size have been established under binary-valued measurements \cite{Tognetti_1968_An, Kadane_1971_Optimal, Zhai_2013_Dynamic}, and under known symmetric distributions of continuous observations \cite{Castanon_1995_Optimal}. In this paper, however, we focus on the sequential setting and general composite hypothesis case.
	
Sequential tests for hypothesis testing problems have attracted much attention since Wald's pioneering work on sequential analysis \cite{Wald_1947_Sequential} due to their property of reaching a decision at a much earlier stage than would be possible with fixed-size tests. Wald established the Sequential Probability Ratio Test (SPRT) for a binary hypothesis testing of a single process. Under the simple hypothesis case, the SPRT is optimal in terms of minimizing the expected sample size under given type $I$ and type $II$ error probability constraints. Various extensions for M-ary hypothesis testing and testing composite hypotheses were studied in \cite{Schwarz_1962_Asymptotic, Lai_1988_Nearly, Pavlov_1990_Sequential, Tartakovsky_2002_Efficient, Draglin_1999_Multihypothesis} for a single process. In these cases, asymptotically optimal performance can be obtained as the error probability approaches zero. In this paper, we focus on asymptotically optimal strategies with low computational complexity for sequential search of a target over multiple processes. Different models considered the case of searching for targets without constraints on the probing capacity, whereas all processes are probed at each given time (i.e., $K=M$, which is a special case of the setting considered in this paper) \cite{Dragalin_1996_Simple, Tartakovsky_2002_Efficient, song2016sequential, song2017asymptotically}.

Since the decision maker can choose which cells to probe, the anomaly detection problem has a connection with the classical sequential experimental design problem first studied by Chernoff \cite{Chernoff_1959_Sequential}. Compared with the classical sequential hypothesis testing pioneered by Wald~\cite{Wald_1947_Sequential} where the observation model under each hypothesis is predetermined, the sequential design of experiments has a control aspect that allows the decision maker to choose the experiment to be conducted at each time.
Chernoff has established a \emph{randomized} strategy, referred to as the Chernoff test which is asymptotically optimal as the maximum error probability diminishes. Chernoff's results were proved for a finite number of states of nature, and in\cite{albert1961sequential} Albert extended Chernoff's results to allow for an infinity of states of nature.
More variations and extensions of the problem and the Chernoff test were studied in~\cite{Bessler_1960_Theory, Nitinawarat_2013_Controlled, nitinawarat2015controlled, Naghshvar_2013_Active, naghshvar2013sequentiality, cohen2015active,deshmukh2019sequential}.
In particular, when the distributions under both normal and abnormal states are completely known under the anomaly detection setting considered here, a modification of the randomized Chernoff test applies and achieves asymptotic optimality~\cite{Nitinawarat_2013_Controlled}.
In our previous work \cite{cohen2015active}, we have shown that a simpler deterministic algorithm applies and obtains the same asymptotic performance, with better performance in the finite sample regime. A modified algorithm has been developed recently in \cite{egan2017fast} for spectrum scanning with time constraint.
In this paper, however, we consider the composite hypothesis case, which is not addressed in \cite{Nitinawarat_2013_Controlled, cohen2015active, egan2017fast}.

In \cite{vaidhiyan2015learning}, searching over Poisson point processes with unknown rates has been investigated and asymptotic optimality has been established when a single location is probed at a time. The policy in \cite{vaidhiyan2015learning} implements a randomized selection rule and also requires to dedicate a linear order of time for exploring the states of all processes. In our model, however, we consider general distributions (with disjoint parameter spaces) and show that deterministic selection rule, with bounded exploration time achieves asymptotic optimality. This result also extends a recent asymptotic result obtained in \cite{nitinawarat2017universal} for non-parametric detection when distributions are restricted to a finite observation space (in contrast to the general continuous valued observations considered here), where asymptotic optimality was shown when the distribution under the null hypothesis is known, a single location is probed at a time, and a logarithmic order of time is used for exploration. In \cite{cohen2015asymptotically}, the problem of detecting abnormal processes over densities that have an unknown parameter was considered, where the process states are independent across cells (in contrast to the problem considered in this paper, in which there is a fixed number of abnormal processes). The objective was to minimize a cost function in the system occurred by abnormal processes, which does not capture the objective of minimizing the detection delay considered here.

Another set of related works is concerned with sequential detection over multiple independent processes~\cite{Li_2009_Restless, Zhao_2010_Quickest, Lai_2011_Quickest, Malloy_2012_Quickest, Tajer_2013_Quick,  Caromi_2013_Fast, malloy2014sequential, cohen2015asymptotically, heydari2016quickest}. In particular, in~\cite{Lai_2011_Quickest}, the problem of identifying the first abnormal sequence among an infinite number of i.i.d. sequences was considered. An optimal cumulative sum (CUSUM) test has been established under this setting. Further studies on this model can be found in~\cite{Malloy_2012_Quickest, Tajer_2013_Quick, malloy2014sequential}. While the objective of finding rare events or a single target considered in~\cite{Lai_2011_Quickest, Malloy_2012_Quickest, Tajer_2013_Quick, malloy2014sequential} is similar to that of this paper, the main difference is that in~\cite{Lai_2011_Quickest, Malloy_2012_Quickest, Tajer_2013_Quick, malloy2014sequential} the search is done over an infinite number of i.i.d processes, where the state of each process (normal or abnormal) is independent of other processes, resulting in open-loop search strategies, which is fundamentally different from the setting in this paper.

Other recent studies include searching for a moving Markovian target\cite{leahy2016always}, and searching for correlation structures of Markov networks \cite{heydari2016quickest_Markov}.

Finally, we point out that our setup is different from the change point detection setup. Our model is suitable to cases where a system has already raised an alarm for event (based on change point detection, for instance), but the location of the event is unknown and needs to be located.
 
\section{System Model and Problem Statement}
\label{sec:system}
	
We consider the problem of detecting a target located in one of $M$ cells quickly and reliably. An extension to detecting multiple targets is discussed in Sec.~\ref{ssec:general}. If the target is in cell $m$, we say that hypothesis $H_m$ is true.
The \emph{a priori} probability that $H_m$ is true is denoted by $\pi_m$, where $\sum_{m=1}^{M}{\pi_m}=1$. To avoid trivial solutions, it is assumed that $0<\pi_m<1$ for all $m$.

We focus on the composite hypothesis case, where the observation distribution has an unknown parameter (or a vector of unknown parameters).
 Let $\theta_m$ be the unknown parameter that specifies the observation distribution of cell $m$. The vector of unknown parameters is denoted by $\boldsymbol{\theta} = (\theta_1 \ldots \theta_M)$.
At each time, only $K$ ($1\leq K\leq M$) cells can be observed. When cell $m$ is observed at time~$n$, an observation $y_m(n)$ is drawn independently from a common density $f\left(y|\theta_m\right)$, $\theta_m\in\Theta$, where $\Theta \subset \mathbb{R}$ is the parameter space for all cells.

If the target is not located in cell $m$, then $\theta_m \in\Theta^{(0)}$; otherwise, $\theta_m \in(\Theta\backslash\Theta^{(0)})$.
The overall parameter space is the Cartesian product $\Theta^M$.
Thus, under hypothesis $H_m$, the true vector of parameters $\boldsymbol{\theta} \in \Theta_m \subset \Theta^M$, where
\begin{center}
$\Theta_m = \{\boldsymbol{\theta}: \theta_i \in \Theta^{(0)}, \forall i \neq m,  \theta_m \in \Theta\backslash\Theta^{(0)}\}$.
\end{center}
Let $\Theta^{(0)}$, $\Theta^{(1)}$ be disjoint subsets of $\Theta$, where $I=\Theta\backslash(\Theta^{(0)}\cup\Theta^{(1)})\neq\emptyset$ is an indifference region\footnote{The assumption of an indifference region is widely used in the theory of sequential composite hypothesis testing to derive asymptotically optimal performance. Nevertheless, in some cases this assumption can be removed. For more details, the reader is referred to \cite{Lai_1988_Nearly}.}. When $\theta^{(1)}\in I$, the detector is indifferent regarding the location of the target. Hence, there are no constraints on the error probabilities for all $\theta\in I$. Shrinking $I$ increases the sample size.
We also assume that $\Theta^{(0)}$, $\Theta^{(1)}$ are open sets.
Let $\mathbf{P}_m$ be the probability measure under hypothesis $H_m$ and $\E_m$ be the operator of expectation with respect to the measure $\mathbf{P}_m$.
	
We define the stopping rule $\tau$ as the time when the decision maker finalizes the search by declaring the location of the target.
\footnote{We point out that it is assumed that the target exists with probability 1. Our model is suitable to cases where a security system has already raised an alarm for event (based on change point detection, for instance), but the location of the event is unknown and need to be located.}
Let $\delta\in\left\{1, 2, ..., M\right\}$ be a decision rule, where $\delta=m$ if the decision maker declares that $H_m$ is true. Let $\phi(n)\in\left\{1, 2, ..., M\right\}^K$ be a selection rule indicating which $K$ cells are chosen to be observed at time~$n$. The time series vector of selection rules is denoted by $\boldsymbol\phi=(\phi(n), n=1, 2, ...)$. Let $\mathbf{y}_{\phi(n)}(n)$ be the vector of observations obtained from cells $\phi(n)$ at time~$n$ and $\mathbf{y}(n)=\left\{\phi(t), \mathbf{y}_{\phi(t)}(t)\right\}_{t=1}^n$ be the set of all cell selections and observations up to time~$n$. A deterministic selection rule $\phi(n)$ at time~$n$ is a mapping from $\mathbf{y}(n-1)$ to $\left\{1, 2, ..., M\right\}^K$. A randomized selection rule $\phi(n)$ is a mapping from $\mathbf{y}(n-1)$ to a probability mass function over $\left\{1, 2, ..., M\right\}^K$. An admissible strategy $\Gamma$ for the anomaly detection problem is given by the tuple $\Gamma=(\tau, \delta, \boldsymbol\phi)$.
	
We adopt a Bayesian approach as in~\cite{Wald_1947_Sequential, Chernoff_1959_Sequential, Lai_1988_Nearly, nitinawarat2012controlled} by assigning a cost of $c$ for each observation and a loss of $1$ for a wrong declaration. Let $P_e(\Gamma)=\sum_{m=1}^{M}{\pi_m\alpha_m(\Gamma)}$ be the probability of error under strategy $\Gamma$, where $\alpha_m(\Gamma)=\mathbf{P}_m(\delta\neq m|\Gamma)$ is the probability of declaring $\delta\neq m$ when $H_m$ is true. Let $\mathbf{E}(\tau|\Gamma)=\sum_{m=1}^{M}{\pi_m\E_m(\tau|\Gamma)}$ be the average detection delay under $\Gamma$. The Bayes risk under strategy $\Gamma$ when hypothesis $H_m$ is true is given by:
$\bea{l}
\displaystyle R_m(\Gamma)\triangleq\alpha_m(\Gamma)+c\E_m(\tau|\Gamma) \;.
\ena$
Note that $c$ represents the ratio of the sampling cost to the cost of wrong detections. The average Bayes risk is given by:
\begin{center}
$\bea{l}
\displaystyle R(\Gamma)=\sum_{m=1}^{M}\pi_m R_m(\Gamma)=P_e(\Gamma)+c\mathbf{E}(\tau|\Gamma) \;.
\ena$
\end{center}
The objective is to find a strategy $\Gamma$ that minimizes the Bayes risk $R(\Gamma)$:
\beq\label{eq:Bayes_formulation1}
\displaystyle\inf_{\Gamma} \;\; R(\Gamma) \;,
\eeq
where the infimum is taken over all randomized and deterministic selection rules.
\begin{definition}
Let $R^*$ be the solution of (\ref{eq:Bayes_formulation1}). We say that strategy $\Gamma$ is asymptotically optimal if
\beq
\displaystyle\lim_{c\rightarrow 0}\frac{R(\Gamma)}{R^*}=1.
\eeq
We note that if the strategy that attains inf does not exist, the definition of the first order asymptotic optimality would be:
\beq
\displaystyle\lim_{c\rightarrow 0}\frac{R(\Gamma)}{\inf_{\Gamma} R(\Gamma)}=1.
\eeq
A shorthand notation $f\sim g$ will be used to denote $\lim_{c\rightarrow 0}f/g=1$.
\end{definition}

A dual formulation (i.e., a frequentist approach) of the problem is to minimize the sample complexity subject to an error constraint $\alpha$, i.e.,:
\beq\label{eq:Frequentist_formulation}
\displaystyle \inf_{\Gamma} \mathbb{E}_m(\tau|\Gamma) , \; \; s.t. \; \; P_e(\Gamma) \leq \alpha \; \; as \; \; \alpha \rightarrow 0
\eeq

In Section \ref{sec:DS} we develop an asymptotically optimal Deterministic Search (DS) algorithm for solving (\ref{eq:Bayes_formulation1}) and (\ref{eq:Frequentist_formulation}).

\subsection{Notations}
\label{notations}
	
We provide next notations that will be used throughout the paper. Let
\beq
\label{eq:unconstrained_MLE}
\displaystyle\hat{\theta}_m(n)\triangleq\arg\max_{\theta\in\Theta}{f\left(\mathbf{\bar{y}}_m(n)|\theta\right)}
\eeq
be the maximum likelihood estimate (MLE) of the parameter over the parameter space $\Theta$ (i.e., unconstrained MLE) at cell $m$, where $\mathbf{\bar{y}}_m(n)=\left(y_m(r_1), ..., y_m(r_{k(n)})\right)$ is the vector of $k(n)$ observations (indicated by times $r_1, ..., r_{k(n)}$) collected from cell $m$ up to time $n$. Regularity conditions for consistency of the MLE are given in App. \ref{ssec:extending}. \\
Let: 
\begin{center}
$\displaystyle\hat{\theta}_m^{(0)}(n)\triangleq\arg\max_{\theta\in\Theta^{(0)}}{f\left(\mathbf{\bar{y}}_m(n)|\theta\right)}$,
\end{center}
\begin{center}
$\displaystyle\hat{\theta}_m^{(1)}(n)\triangleq\arg\max_{\theta\in\Theta \setminus \Theta^{(0)}}{f\left(\mathbf{\bar{y}}_m(n)|\theta\right)}$
\end{center}
be the MLE for cell $m$ to be in normal or abnormal state, respectively.

Let $\mathbf{1}_m(n)$ be the indicator function, where $\mathbf{1}_m(n)=1$ if cell $m$ is observed at time~$n$, and $\mathbf{1}_m(n)=0$ otherwise. \\
We now propose two optional statistics. Let\\
\beq
\label{eq:MGLLR_L}
\displaystyle S_{m,LGLLR}^{(r)}(n) \buildrel \Delta \over = \sum\limits_{t = 1}^n {{{\bf{1}}_m}} (t)\log {{f({y_m}(t)|{{\hat \theta }_m}(n))} \over {f({y_m}(t)|{{\hat \theta_m }^{(r)}}(n))}}
\eeq
be the sum of Local Generalized Log-Likelihood Ratio (LGLLR) of cell $m$ at time~$n$ used to reject hypothesis $r$ (for $r=0, 1$) regarding its state.
We refer to the statistics as \textit{local} since it uses the observations from cell $m$ solely. In Section \ref{ssec:general} we will define a statistics measure that uses observations from multiple cells, referred to as Multi-process Generalized Log-Likelihood Ratio (MGLLR).
The LGLLR statistics is inspired by the Generalized Likelihood Ratio (GLR) statistics used for sequential tests, first studied by Schwartz \cite{Schwarz_1962_Asymptotic} for a one parameter exponential family, who assigned a cost of $c$ for each observation and a loss function for wrong decisions. A refinement was studied by Lai \cite{Lai_1988_Nearly, Lai_1994_Nearly}, who set a time-varying boundary value. Lai showed that for a multivariate exponential family this scheme asymptotically minimizes both the Bayes risk and the expected sample size subject to error constraints as $c$ approaches zero \cite{Lai_1994_Nearly}.   

The second statistics that we propose to use is obtained by replacing the parameter for the $k$th observation with the estimator $\hat \theta_m(k-1)$ built upon samples $\mathbf{\tilde{y}}_m(n)=\left(y_m(r_1), ..., y_m(r_{k-1(n)})\right)$. The statistics is given by:
\beq
\label{eq:MALLR_L}
\displaystyle S_{m,LALLR}^{(r)}(n) \buildrel \Delta \over = \sum\limits_{t = 1}^n {{{\bf{1}}_m}} (t)\log {{f({y_m}(t)|{{\hat \theta }_m}(t-1))} \over {f({y_m}(t)|{{\hat \theta_m }^{(r)}}(n))}},
\eeq
which we refer to as the sum of Local Adaptive Log Likelihood Ratio (LALLR).
The LALLR statistics is inspired by the Adaptive Likelihood Ratio (ALR) statistics used for sequential tests, first introduced by Robbins and Siegmund \cite{Robbins_1974_Expected} to design power-one sequential tests.  Pavlov used it to design asymptotically (as the error probability approaches zero) optimal (in terms of minimizing the expected sample size subject to error constraints) tests for composite hypothesis testing of the multivariate exponential family \cite{Pavlov_1990_Sequential}. Tartakovsky established asymptotically optimal performance for a more general multivariate family of distributions \cite{Tartakovsky_2002_Efficient}. \\
The advantage of using the LALLR statistics, is that it enables us to upper-bound the error probabilities of the sequential test by using simple threshold settings. Thus, implementing the LALLR is much simpler than implementing the LGLLR. The disadvantage of using the LALLR is that poor early estimates (for small number of observations) can never be revised even though one has a large number of observations. A numerical comparison for the performance of the two statistics is presented in Section \ref{ssec:experiments_amllr_gmllr}.

Finally,
\begin{center}
$D(x||z)\triangleq\mathbf{E}_{f(y(n)|x)}\left(\log\frac{f(y(n)|x)}{f(y(n)|z)}\right)$
\end{center}
denotes the Kullback–Leibler (KL) divergence between two distributions, $f(y(n)|x), f(y(n)|z)$.
	
\section{A Low-Complexity Deterministic Search (DS) Algorithm}
\label{sec:DS}

Sequential detection problems involving multiple processes are POMDP \cite{Castanon_1995_Optimal}. As a result, computing optimal search policies is intractable in general. In this section we present the Deterministic Search (DS) algorithm, which has low complexity (linear with the number of processes) used for solving the anomaly detection problem asymptotically as the error approaches zero. Both proposed statistics (LGLLR and LALLR) can be used in the implementation of the algorithm.

We start by analyzing the case where no additional side information on the process states is known in Section \ref{ssec:uncorrelated}. Then, in Section \ref{ssec:known}, we consider the case in which the parameter value under the null hypothesis is known and equal for all normal processes. In this case we show analytically the gain achieved in the detection time, by utilizing the side information on the normal state. Finally, in Section \ref{ssec:general}, we consider the case where the parameter value under the null hypothesis is unknown, but is known to be equal for all normal processes.     

\subsection{Anomaly Detection Without Side Information}
\label{ssec:uncorrelated}
We assume that $K=1$ as widely assumed in dynamic search problems for purposes of analysis (e.g., \cite{Zigangirov_1966_Problem, Klimko_1975_Optimal, Dragalin_1996_Simple, Stone_1971_Optimal, cohen2014optimal, vaidhiyan2015learning, nitinawarat2017universal}). In Section \ref{ssec:general} we discuss the implementation under more general settings. 
We also assume that the parameter space is finite, and we assume a large-scale system where $M>>1$ so that $D(\theta^{(0)}||\theta^{(1)})/D(\theta^{(1)}||\theta^{(0)})<M-1$ for all $\theta^{(0)}\in\Theta^{(0)}$, $\theta^{(1)}\in\Theta^{(1)}$.
Let $\mathcal{H}_1(n)=\left\{m\;:\; \hat{\theta}_m(n)\nin\Theta^{(0)}\right\}$ be the set of cells whose MLEs lie outside $\Theta^{(0)}$ at time $n$ with cardinality $|\mathcal{H}_1(n)|=N_{H_1}(n)$.
 Let $S_m^{(r)}(n)$ be the $S_{m,LALLR}^{(r)}(n)$ or $S_{m,LGLLR}^{(r)}(n)$ statistics defined in Section \ref{notations}. The DS algorithm has a structure of exploration and exploitation epochs. We start by addressing the Bayesian formulation, and we describe the DS algorithm with respect to time index $n$. \vspace{0.1cm}\\
\begin{enumerate}
  \item (Exploration phase:) If $N_{H_1}(n)\neq 1$, then probe the cells one by one in a round-robin manner, i.e., $\phi(n)=\left[\left(\phi(n-1)+1\right)\mod M\right]$ and go to Step $1$ again. Otherwise, go to Step $2$.
  \item (Exploitation phase:) Update $\hat{\theta}_m(n)$ for all $m=1, ..., M$, and let $\hat{m}(n)=\left\{m\;:\; \hat{\theta}_m(n)\nin\Theta^{(0)}\right\}$ be the index of the cell whose MLE lies outside $\Theta^{(0)}$ at time $n$ (note that this cell is unique at the exploitation phase). Probe cell $\phi(n)=\hat{m}(n)$ and go to Step 3.
  \item (Sequential testing:) Update $S_{\phi(n)}^{(0)}(n)$ based on the last observation. If $S_{\hat{m}(n)}^{(0)}(n)\geq -\log c$ stop the test and declare $\delta=\hat{m}(\tau)$ as the location of the target. Otherwise, go to Step 1.\vspace{0.2cm}
\end{enumerate}

Note that the selection rule constructed by Steps 1, 2 is deterministic and dynamically updated based on the current value of the MLEs.
The proposed DS algorithm is intuitively satisfying. Consider first the simple hypothesis case (where asymptotic optimality was shown in \cite{cohen2015active}), in which $\theta^{(0)}, \theta^{(1)}$ are assumed known. When $K=1$ and $M>>1$, the DS algorithm selects at each time the cell with the largest sum log likelihood ratio. The intuition behind this selection rule is that $D(\theta^{(1)}||\theta^{(0)})$ and $D(\theta^{(0)}||\theta^{(1)})/(M-1)$ determine, respectively, the rates at which the state of the cell with the target and the states of the $M-1$ cells without the target can be accurately inferred. Since $M>>1$ such that $D(\theta^{(0)}||\theta^{(1)})/D(\theta^{(1)}||\theta^{(0)})<M-1$ for all $\theta^{(0)}\in\Theta^{(0)}$, $\theta^{(1)}\in\Theta^{(1)}$, the DS algorithm aims at identifying the cell with the target (which is equivalent to probe the most likely abnormal process as implemented during the exploitation phase). When handling the composite hypothesis case and $\theta^{(1)}$ is unknown, the selection rule dedicates an exploration phase for estimating the parameter and adjusts the estimated KL divergences dynamically. Since the parameter spaces are disjoint, the exploration phase yields an estimate for the location of the abnormal process (i.e., the cell whose MLE lies outside $\Theta^{(0)}$). The exploitation phase keeps taking samples until $S_{\hat{m}(n)}^{(0)}(n)\geq -\log c$ first occurs to ensure a sufficiently accurate decision, i.e., error probability of order $O(c)$ as shown in the analysis. \\

\begin{theorem}
\label{th:DS_optimality_policy_uncorrelated}
Assume that the DS algorithm is implemented under the anomaly detection setting described in this section. Let $R^*$ and $R(\Gamma)$ be the Bayes risks under the DS algorithm and any other policy $\Gamma$, respectively. Then, the following statements hold:\vspace{0.1cm}\\
1) (\emph{Finite sample error bound:}) The error probability is upper bounded by $(M-1)c$ for all $c$.\vspace{0.1cm}\\
2) (\emph{Asymptotic optimality:}) The Bayes risk satisfies:
\begin{center}
$\bea{l}
\displaystyle R^* \;\sim\; \frac{-c\log c}{D(\theta^{(1)})}\;\sim\;\inf_{\Gamma}\;{R(\Gamma)} \;\;\; \mbox{as} \;\;\; c\rightarrow 0 \;,\vspace{0.1cm}
\ena$
\end{center}
\vspace{0.1cm} \hspace{0.3cm} where $\displaystyle D(\theta^{(1)}) \triangleq \min_{\varphi \in \Theta^{(0)}} D(\theta^{(1)}||\varphi)$. \vspace{0.1cm} \\
3) (\emph{Bounded exploration time:}) The total expected time spent during the exploration phase (i.e., Step 1 in the DS algorithm) is $O(1)$.\vspace{0.2cm}
\end{theorem}
The proof is given in Appendix \ref{app:proof_policy_uncorrelated}.

We point out that bounded exploration time of the DS algorithm is of particular significance. It is in sharp contrast with the logarithmic order of exploration time commonly seen in active search strategies (see, for example, \cite{cohen2015asymptotically, nitinawarat2017universal}).

\subsection{Anomaly Detection under a Known Model of Normality}
\label{ssec:known}
	
Here, we assume that the parameter under null hypothesis $\theta=\theta^{(0)}\in\Theta^{(0)}$ is known, and equal for all empty cells, where $\Theta^{(0)}$ is an open set that contains $\theta^{(0)}$. This setting models many anomaly detection situations, in which the distribution of the observations under a normal state is known, while there is uncertainty in the distribution under an abnormal state. 
To utilize this information, we adjust the LALLR statistics used to reject hypothesis $H_0$ as follows: 
\beq
\label{eq:ALLR_Nomral}
\displaystyle \widetilde{S}_{m,LALLR}^{(0)}(n) \buildrel \Delta \over = \sum\limits_{t = 1}^n {{{\bf{1}}_m}} (t)\log {{f({y_m}(t)|{{\hat \theta }_m}(t-1))} \over {f({y_m}(t)|{{ \theta }^{(0)}})}}.
\eeq
We define $\widetilde{S}_{m,LGLLR}^{(0)}(n)$ similarly.

In the following theorem we establish a finite-sample upper bound on the error probability and prove asymptotic optimality of the algorithm for the Bayesian formulation using the adjusted LALLR statistics, where only $O(1)$ order of time is spent during the exploration phase. The proof is given in App. \ref{app:proof_policy1}. \vspace{0.2cm}
	
\begin{theorem}
\label{th:DS_optimality_policy1}
Assume that the DS algorithm is implemented under the anomaly detection setting described in this section, using the adjusted LALLR statistics. Let $R^*$ and $R(\Gamma)$ be the Bayes risks under the DS algorithm and any other policy $\Gamma$, respectively. Then, the following statements hold:\vspace{0.1cm}\\
1) (\emph{Finite sample error bound:}) The error probability is upper bounded by $(M-1)c$ for all $c$.\vspace{0.1cm}\\
2) (\emph{Asymptotic optimality:}) The Bayes risk satisfies:
\begin{center}
$\bea{l}
\displaystyle R^* \;\sim\; \frac{-c\log c}{D(\theta^{(1)}||\theta^{(0)})}\;\sim\;\inf_{\Gamma}\;{R(\Gamma)} \;\;\; \mbox{as} \;\;\; c\rightarrow 0 \;.\vspace{0.1cm}
\ena$
\end{center}
3) (\emph{Bounded exploration time:}) The total expected time spent during the exploration phase (i.e., Step 1 in the DS algorithm) is $O(1)$.\vspace{0.2cm}
\end{theorem}

We point out that the side information on the true null hypothesis strengthens the algorithm performance. The improvement in the performance is clearly seen by the fact that $D(\theta^{(1)}||\theta^{(0)}) \geq D(\theta^{(1)})$. Hence, the risk in Theorem \ref{th:DS_optimality_policy1} is smaller then the risk in Theorem \ref{th:DS_optimality_policy_uncorrelated}. Note also that in this setting we do not restrict $\Theta^{(0)}$ to be a singleton set (the parameter still lies in an open set). The side information is utilized when constructing the statistics in (\ref{eq:ALLR_Nomral}).

For the frequentist formulation, in step 3 of the DS algorithm (i.e., sequential testing step) we define the threshold as $a$, i.e., if $S_{\hat{m}(n)}^{(0)}(n)\geq a$ we stop the test and declare $\delta=\hat{m}(\tau)$ as the location of the target. We now present Theorem \ref{th:DS_optimality_policy1_2}, which claims that the DS algorithm is first order asymptotically optimal in the sense of criterion (\ref{eq:Frequentist_formulation}). The proof is given in App.\ref{app:proof_policy1_2}. \vspace{0.2cm}

\begin{theorem}
\label{th:DS_optimality_policy1_2}
Assume that the DS algorithm is implemented under the anomaly detection setting described in this section, using the adjusted LALLR statistics. We define the class of tests \vspace{0.1cm}:
\begin{center}
$\vspace{0.1cm} C(\alpha)= \{\Gamma: P_e(\Gamma) \leq \alpha)$.
\end{center}
Let $\mathbb{E}_m(\tau|\Gamma^*)$ and $\mathbb{E}_m(\tau|\Gamma)$ be, the detection time under the DS algorithm, and any other policy, respectively.  
Then, the following statement holds for each $m=1,\ldots, M$:
\beq
\begin{split}
\displaystyle \mathbb{E}_m(\tau|\Gamma^*) \;\sim\;  \inf_{\Gamma \in C(\alpha)} \mathbb{E}_m(\tau|\Gamma) =  \frac{|\log \alpha|}{D(\theta^{(1)}||\theta^{(0)})} & (1+o(1)), \; \; \\
& as \; \;\alpha \rightarrow 0,
\end{split}
\eeq
and $\Gamma^*\in C(\alpha)$.
\end{theorem}

\subsection{Anomaly Detection under Identical Parameter for All Normal Cells}
\label{ssec:general}
Next, we consider the case where both parameter values under normal and abnormal states $\theta^{(0)}$ and $\theta^{(1)}$ are unknown. However, it is known that the unknown parameter is identical for all normal cells. Therefore, under hypothesis $H_m$, the true vector of parameters satisfy $\boldsymbol{\theta} \in \Theta_m \subset \Theta^M$, where 
\begin{center}
$\Theta_m = \{\boldsymbol{\theta}: \theta_i = \theta^{(0)} \in \Theta^{(0)}, \forall i \neq m,\;  \theta_m = \theta^{(1)} \in \Theta\backslash\Theta^{(0)}\}$.
\end{center}
Note that in contrast to section \ref{ssec:uncorrelated} where observations from cell $m$ does not contribute any information about the parameter value of cell $r$, for $m\neq r$, here the additional side information allows us to estimate the true value of $\theta^{(0)}$ consistently using observations from each normal cell. Specifically, let $\mathbf{y}_{\Theta^{(0)}}(n), \mathbf{y}_{\Theta\setminus \Theta^{(0)}}(n)$ be the set of all the observations collected from the cells whose MLEs lie inside $\Theta^{(0)}$ (i.e., $\hat{\theta}_m(n)\in\Theta^{(0)}$), and inside $\Theta\setminus \Theta^{(0)}$ (i.e., $\hat{\theta}_m(n)\in\Theta\setminus\Theta^{(0)}$) at time $n$, respectively. The global MLE of $\theta^{(0)}$ is computed based on the observations from all the cells which are likely to be empty:
\begin{center}
$\bea{l}
\displaystyle\hat{\theta}^{(0)}(n)\triangleq\arg\max_{\theta\in\Theta}{f\left(\mathbf{y}_{\Theta^{(0)}}(n)|\theta\right)},
\ena$
\end{center}
where the global MLE of $\theta^{(1)}$ is computed based on the observations from all the cells which are likely to contain the target:
\begin{center}
$\bea{l}
\displaystyle\hat{\theta}^{(1)}(n)\triangleq\arg\max_{\theta\in\Theta}{f\left(\mathbf{y}_{\Theta\setminus\Theta^{(0)}}(n)|\theta\right)}.
\ena$
\end{center}
Intuitively, as more observations are collected from all cells, only the MLE at the cell that contains the target is likely to lie inside $\Theta\setminus \Theta^{(0)}$.
Next, we define the statistics accordingly. Let:
\beq
\label{eq:MGLLR}
\displaystyle S_{m,MGLLR}^{(r)}(n) \buildrel \Delta \over = \sum\limits_{t = 1}^n {{{\bf{1}}_m}} (t)\log {{f({y_m}(t)|{{\hat \theta }_m}(n))} \over {f({y_m}(t)|{{\hat \theta }^{(r)}}(n))}}
\eeq
be the sum of Multi-process Generalized Log-Likelihood ratio (MGLLR) of cell $m$ at time $n$ used to reject hypothesis $r$ (for $r = 0,1$) regarding its state.
The modified adaptive statistics is defined by:
\beq  
\label{eq:MALLR}
\displaystyle S_{m,MALLR}^{(r)}(n) \buildrel \Delta \over = \sum\limits_{t = 1}^n {{{\bf{1}}_m}} (t)\log {{f({y_m}(t)|{{\hat \theta }_m}(t-1))} \over {f({y_m}(t)|{{\hat \theta }^{(r)}}(n))}},
\eeq
which we refer to as the sum of Multi-process Adaptive Log-Likelihood Ratio (MALLR)
\footnote{Note that the adaptive LLR statistics and generalized LLR statistics used in sequential composite hypothesis testing of a single process contains a constrained MLE over the alternative parameter space in the denominator (see Section \ref{ssec:experiments_allr} for more details). Here, we use unconstrained MLEs (which are computed over the entire parameter space $\Theta$) in both numerator and denominator, depending on the cells from which the observations were taken. Thus, we refer to this statistics measure as a Multi-process Adaptive/Generalized LLR (MALLR/ MGLLR).}. \\
Let $\mathcal{N}_s=\left\{n_1, n_2, ...\right\}$ be a sequence of time instants, where $O(|\mathcal{N}_s|)$ has a logarithmic order of time, in which the cells are selected in a round-robin manner during the algorithm. Intuitively speaking, the role of $n_1, n_2, ...$, is to explore all the cells to infer the true value of $\theta^{(0)}$ (which is not observed when testing the target cell) during the algorithm. This allows us to use the estimate values of both $\theta^{(0)}$ and $\theta^{(1)}$ when computing the statistics used in the algorithm. We also define $m^{(i)}(n)$ for $i=1, 2, ..., M-1$ as the index of the cell with the $i^{th}$ smallest sum MALLR $S_j^{(1)}(n)$ for $j\neq\hat{m}(n)$ at time~$n$.
The DS algorithm has a structure of exploration and exploitation epochs. Let $S_m^{(r)}(n)$ be the statistics used in the algorithm which can be the MALLR or MGLLR statistics. Next, we describe the DS algorithm with respect to time index $n$. We describe the algorithm for the general case where multiple processes can be probed at a time ($K\geq 1$), and $D(\theta^{(0)}||\theta^{(1)})/D(\theta^{(1)}||\theta^{(0)})<M-1$ does not necessarily hold. \vspace{0.1cm}\\
\begin{enumerate}
  \item (Exploration phase 1:) Exploration phase 1 is similar to the exploration phase described in Section \ref{ssec:known}. If $N_{H_1}(n)\neq 1$, then cells are probed one by one in a round-robin manner. Otherwise, go to Step $2$.
  \item (Exploration phase 2:) If $n\in\mathcal{N}_s$, the cells are probed one by one in a round-robin manner. Otherwise, if $N_{H_1}(n)\neq 1$, go to Step $1$. Otherwise, go to Step $3$.
  \item (Exploitation phase:) Update $\hat{\theta}_m(n)$ for all $m=1, ..., M$, and let $\hat{m}(n)=\left\{m\;:\; \hat{\theta}_m(n)\nin\Theta^{(0)}\right\}$ be the index of the cell whose MLE lies outside $\Theta^{(0)}$ at time $n$ (note that this cell is unique at the exploitation phase).
      Then, probe cells $\phi(n)$ which are given by:\footnote{Assume that $K<M$. Otherwise, all cells are probed.}
	\beq
	\label{eq:selection_policy2}
	\displaystyle \phi(n)=
	\begin{cases} \left(\hat{m}(n), m^{(1)}(n), m^{(2)}(n), ..., m^{(K-1)}(n)\right) \;, \vspace{0.2cm}\\ \hspace{0.5cm}
		\mbox{if \;$D(\hat{\theta}^{(1)}(n)||\hat{\theta}^{(0)}(n))\geq\frac{D(\hat{\theta}^{(0)}(n)||\hat{\theta}^{(1)}(n))}{(M-1)}$ \vspace{0.2cm}}\\\mbox{ \hspace{2.0cm}
			and $n\nin\left\{\mbox{exploration times}\right\}$
		}
		\vspace{0.0cm}\\
		\left(m^{(1)}(n), m^{(2)}(n), ..., m^{(K)}(n)\right)\;, \vspace{0.2cm}\\ \hspace{0.5cm}
		\mbox{if \;$D(\hat{\theta}^{(1)}(n)||\hat{\theta}^{(0)}(n))<\frac{D(\hat{\theta}^{(0)}(n)||\hat{\theta}^{(1)}(n))}{(M-1)}$ \vspace{0.2cm}}\\\mbox{ \hspace{2.0cm}
			and $n\nin\left\{\mbox{exploration times}\right\}$
		}
	\end{cases}
	\eeq
and go to Step 4.
  \item (Sequential testing:) Update the sum MALLRs based on the last observations. If $S_{\hat{m}(n)}^{(0)}(n)+S_{m^{(1)}(n)}^{(1)}(n)\geq -\log c$ stop the test and declare $\delta=\hat{m}(\tau)$ as the location of the target. Otherwise, go to Step 1. \vspace{0.2cm}
\end{enumerate}

The proposed DS algorithm under the general setting is intuitively satisfying. Since both $\theta^{(0)}, \theta^{(1)}$ might be unknown, the selection rule dedicates exploration phases $1, 2$ for estimating the parameters and adjusts the estimated KL divergences dynamically. Since the parameter spaces are disjoint, exploration phase $1$ yields an estimate for the location of the abnormal process (i.e., the cell whose MLE lies outside $\Theta^{(0)}$). The exploitation phase keeps taking samples until $S_{\hat{m}(n)}^{(0)}(n)+S_{m^{(1)}(n)}^{(1)}(n)\geq -\log c$ first occurs, i.e., to ensure a sufficiently accurate decision. We show in the appendix that this stopping rule achieves error probability of order $O(c)$ when the parameters are known under both normal and abnormal states, and polynomial decay with time is achieved under the general composite hypothesis testing setting (though only consistency can be shown, where asymptotic optimality still remains open in the general setting), which motivates the design of the stopping rule.

In the theorem below, we prove the consistency of the DS algorithm using the MALLR statistics. The proof and regularity assumptions are given in Appendix \ref{app:proof_policy4}. \vspace{0.2cm}

\begin{theorem}
\label{th:DS_generelized_policy}
Assume that the DS algorithm is implemented under the anomaly detection setting described in this section. Assume also that the parameters $\theta^{(0)}, \theta^{(1)}$ can take a finite number of values (where the observations are still continuous). Let $H_m$ be true hypothesis. Then, $\hat{m}(n)\rightarrow m$ as $c\rightarrow 0$, and the error probability decays polynomially with $-\log c$.\vspace{0.2cm}
\end{theorem}

It should be noted that the expected detection time is of order $O(-\log c)$. Therefore, Theorem \ref{th:DS_generelized_policy} implies that the error probability decays polynomially with the expected detection time. We point out that establishing asymptotic optimality for $K>1$ remains open. In this case, at each time slot the statistics is based on a mixed of samples from cells that contain the target and from cells that do not contain the target. As a result, bounding the error probability by $O(c)$ while achieving the asymptotically optimal detection time is much more complex. \\
In Figure \ref{fig:sim1} we present simulation results, demonstrating strong performance of the DS algorithm under the setting considered in this section.
The sum MALLRs use the exact values of $\theta^{(0)}, \theta^{(1)}$ when they are known, and the MLEs of $\theta^{(0)}, \theta^{(1)}$ when they are unknown. Although theoretical asymptotic optimality remains open when $\theta^{(0)}, \theta^{(1)}$ are unknown (and $\theta^{(0)}$ is identical for all normal cells), it can be seen by simulations that the DS algorithm nearly achieves asymptotically optimal performance in this case as well (since it approaches the performance of the DS algorithm when $\theta^{(0)}, \theta^{(1)}$ are known).\vspace{0.2cm}
	
\begin{figure}[htbp]
\centering \epsfig{file=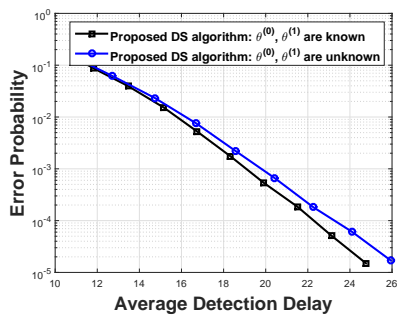,
width=0.35\textwidth}
\caption{The error probability as a function of the average detection delay under the proposed DS algorithm. A case of Laplace distributions with parameters $\theta^{(0)} = 0, \theta^{(1)} = 1$ under normal and abnormal states, respectively, with $K = 2, M = 5$. We averaged over $10^6$ Monte Carlo runs.}
\label{fig:sim1}
\end{figure}

\begin{remark}
Note that in Sections \ref{ssec:uncorrelated} and \ref{ssec:known} the exploitation phase collects observations from cell $\hat{m}(n)$. As a result, a sufficiently accurate MLE for $\theta^{(1)}$ is computed based on observations collected during the exploitation phase, while exploration phase $2$ is unnecessary.  In the setting considered in this section, however, exploration phase $2$ is required to guarantee a sufficiently accurate estimation of the unknown parameter $\theta^{(0)}$. Specifically, let $N_O(t)$ denote the number of observations that have been collected in exploration phase 2, and let $\tau^{(0)}_{ML}$ be the smallest integer such that $\hat{\theta}^{(0)}(n) = \theta^{(0)}$ for all $n>\tau^{(0)}_{ML}$. 
Then, exploring cells such that $N_O(t)>\frac{2}{I_0} \log(t)$ is met for all $t$ is sufficient to ensure consistency of $\hat{\theta}^{(0)}(n)$, where
\beq
\label{eq:Exploration_rate}
\vspace{0.1cm} \displaystyle I_0 \triangleq \inf_{\theta^{(0)}, \theta \in \Theta : \theta^{(0)}\neq \theta } \sup_{s>0} \left\{ -\log\mathbf{E}_{\sim f(y|\theta^{(0)})}\left[e^{s(-\ell^{(\theta^{(0)}, \theta)})}\right] \right\}
\eeq
is the Legendre-Fenchel transformation of \vspace{0.2cm}
\begin{center}
$\vspace{0.1cm} \displaystyle \ell^{(\theta^{(0)},\theta)} = \log \frac{f(y|\theta^{(0)})}{f(y|\theta)}$. 
\end{center}
Below, we prove the statement (under hypothesis $H_m$ w.l.o.g.): \vspace{0.1cm} \\
$\vspace{0.2cm} \hspace{0.0cm} \displaystyle \mathbf{P}_m \left(\tau_{ML}^{(0)}>n\right) \leq \sum_{t=n}^{\infty} 
\mathbf{P}_m \left(\hat{\theta}^{(0)} (t) \neq \theta^{(0)} \right).$\\
By the definition of $\hat{\theta}^{(0)}(n)$, the event $\hat{\theta}^{(0)}(t) \neq \theta^{(0)}$ implies: 
\beq
\displaystyle\sum_{i=1}^{t}\ell^{(\theta^{(0)}, \widetilde{\theta}(t))}(i)<0,
\eeq
for some $\widetilde{\theta}(t)\neq\theta^{(0)}$, where the index $i$ refers to measurement $i$ taken from cells which are likely to be empty. Since the expected last exit time (say $t'$) from exploration phase 1 is bounded (see Appendix \ref{app:proof_policy1}), applying the Chernoff bound for all $t>t'$ and using the i.i.d property yields: \beq
\bea{l}
\vspace{0.2cm} \hspace{0.0cm} \displaystyle\mathbf{P}_m\left(\sum_{i=1}^{t}\ell^{(\theta^{(0)}, \widetilde{\theta}(t))}(i)<0\right)\\ 
\vspace{0.2cm} \hspace{0.0cm} \displaystyle\leq\displaystyle \min_{s>0} \left\{\mathbf{E}_{\sim f(y|\theta^{(0)})}\left[e^{s(-\ell^{(\theta^{(0)}, \widetilde{\theta}(t))}(i))}\right]\right\}^{N_O(t)} \\
\vspace{0.2cm} \hspace{0.0cm} \displaystyle = \min_{s>0} \left\{e^{-N_O(t)\left(-\log\mathbf{E}_{\sim f(y|\theta^{(0)})}\left[e^{s(-\ell^{(\theta^{(0)}, \widetilde{\theta}(t))}(i))}\right] \right) }  \right\} \\
\vspace{0.0cm} \hspace{0.0cm} \displaystyle = e^{-N_O(t)\left(\sup_{s>0} \left\{ -\log\mathbf{E}_{\sim f(y|\theta^{(0)})}\left[e^{s(-\ell^{(\theta^{(0)}, \widetilde{\theta}(t))}(i))}\right] \right\} \right) }.
\ena
\eeq
Since $ N_O(t) > \frac{2}{I_0} \log(t)$, $\mathbf{E}_{\sim f(y|\theta^{(0)})}[\tau^{(0)}_{ML}] = O(1)$ is satisfied.\vspace{0.2cm}
\end{remark}

\begin{remark}
It should be noted that the proposed DS algorithm can be extended to handle multiple (say $L$) abnormal processes as well. The exploration phase can be implemented in a similar manner until exactly $L$ MLEs lie outside $\Theta^{(0)}$. The exploitation phase will prioritize processes which are likely to be abnormal if the conditions on the first line of (\ref{eq:selection_policy2}) hold. Otherwise, it will prioritize processes which are likely to be normal if the conditions on the second line of (\ref{eq:selection_policy2}) hold. The test terminates once all the abnormal processes are distinguished from the rest $M-L$ normal processes, i.e., when the $L^{th}$ highest sum MALLR among the processes which are likely to be abnormal plus the smallest sum MALLR among the processes which are likely to be normal is greater than $-\log c$.
\end{remark}

\subsection{Comparison with Chernoff's test}
\label{ssec:Chernoff and Albert}
In this section, we discuss the differences between our problem and the classical sequential experimental design problem studied by Chernoff, first presented in\cite{Chernoff_1959_Sequential}. While we presented a deterministic algorithm search, Chernoff  proposed a test with a randomized selection rule. Specifically, let $q=(q_1, ..., q_N)$ be a probability mass function over a set of $N$ available experiments $u=\left\{u_i\right\}_{i=1}^N$ that the decision maker can choose from, where $q_i$ is the probability of choosing experiment $u_i$. For a general M-ary sequential design of experiments problem, the action at time~$n$ under the Chernoff test is drawn from a distribution $q^*(n)=(q^*_1(n), ..., q^*_N(n))$ that depends on the past actions and observations:
	\beq
	\label{eq:selection_Chernoff}
	\displaystyle q^*(n)=\arg\;\max_{q}\;\min_{j\in\mathcal{M}\setminus\left\{\hat{i}(n)\right\}}
	\sum_{u_i}q_i D(p_{\hat{i}(n)}^{u_i}||p_j^{u_i})\;,
	\eeq
	where $\mathcal{M}$ is the set of the $M$ hypotheses, $\hat{i}(n)$ is the MLE of the true hypothesis at time~$n$ based on past actions and observations, and $p_j^{u_i}$ is the observation distribution under hypothesis $j$ when action $u_i$ is taken.\\
Chernoff's results were proved only for a finite number of states of nature (set of possible parameters). Albert \cite{albert1961sequential} extended Chernoff's results to allow for an infinity of states of nature.
Beyond the differences in the deterministic versus randomized selection rules, we will now discuss in details the connection with the model considered by Chernoff and Albert. 
(i) \emph{Violating the positivity assumption on the KL divergence:} The asymptotic optimality of the Chernoff test as shown in~\cite{Chernoff_1959_Sequential, albert1961sequential} requires that under any experiment, any pair of hypotheses are distinguishable (i.e., has positive KL divergence). This assumption does not hold in the anomaly detection settings considered in this paper. For instance, under the experiment of searching the $i^{th}$ cell, the hypotheses of the target being in the $j^{th}$ ($j\neq i$) and the $k^{th}$ ($k\neq i$) cells yield the same observation distribution. In~\cite{Nitinawarat_2013_Controlled}, the authors relaxed this assumption, and developed a modified Chernoff test in order to handle indistinguishable hypotheses under some (but not all) actions. The basic idea of the modified test is to implement an exploration phase with a uniform distribution for a subsequence of time instants that grows logarithmically with time. Although asymptotic optimality was proved under the modeified Chernoff test, its exploration time is unbounded, and affects the finite-time performance. Nevertheless, in this paper we have shown that the DS algorithm achieves asymptotic optimality under both settings in Sections \ref{ssec:uncorrelated}, \ref{ssec:known}, using a bounded exploration time.
\emph{(ii) Utilizing the side information in the anomaly detection setting:} The model in \cite{Chernoff_1959_Sequential, albert1961sequential} can be embedded to the model in Section \ref{ssec:uncorrelated} (with the extension in \cite{Nitinawarat_2013_Controlled} as discussed earlier). This embedding does not contain side information on the parameter values under different hypotheses. The analysis in \cite{Chernoff_1959_Sequential, albert1961sequential} relies on rejecting the alternative hypothesis with respect to the closest alternative. Indeed, the DS algorithm achieves the same asymptotic optimality as in \cite{Chernoff_1959_Sequential}, but with deterministic selection rule, with better finite-time performance as demonstrated in the simulation results. The asymptotically optimal Bayes risk is given in this case by $\sim-c\log c/ \inf_{\varphi \in \Theta^{(0)}} D(\theta^{(1)} || \varphi)$ which matches with the asymptotically optimal performance in \cite{Chernoff_1959_Sequential, albert1961sequential}. 
Asymptotic optimality of the Chernoff test is achieved under the model setting in Section \ref{ssec:known} by embedding the parameter set under $\theta^{(0)}$ to a singleton, and thus the same asymptotic performance can be achieved, where the asymptotically optimal Bayes risk is given in this case by $\sim-c\log c/ D(\theta^{(1)} || \theta^{(0)})$. Indeed, we have shown that the DS algorithm achieves the same asymptotic optimality as in \cite{Chernoff_1959_Sequential, albert1961sequential} in this case. 
However, asymptotic optimality under the Chernoff test remains open in the setting considered in Section \ref{ssec:general}, since it cannot be embedded as in Section \ref{ssec:known}. The asymptotic analysis in \cite{Chernoff_1959_Sequential, albert1961sequential} is established with respect to the entire parameter space (as in Section \ref{ssec:uncorrelated}), while the lower bound on the risk must be developed with respect to the true parameter values that satisfy the side information. Nevertheless, intuitively, one can expect to improve performance by estimating the parameter $\theta^{(0)}$ consistently and improve the detection performance by approaching the performance in Section \ref{ssec:known}. We indeed showed that the DS algorithm achieves consistency in this setting.

Despite the differences between the two models, we extended the randomized Chernoff test for the anomaly detection problem over composite hypotheses as follows.   
We select cells from a uniform distribution at exploration phase until only a single MLE lies outside $\Theta^{(0)}$. Then, the solution of (\ref{eq:selection_Chernoff}) is executed in the exploitation phase.
The randomized test in (\ref{eq:selection_Chernoff}) chooses, at each time, a probability distribution that governs
the selection of the experiment to be carried out at this time. This distribution is obtained by solving a maximin problem so
that the next observation will best differentiate the current MLE of the true hypothesis from its closest alternative, where the distance is measured by the KL divergence. 
It can be shown that when applied to the anomaly detection problem, the solution of (\ref{eq:selection_Chernoff}) works as follows. Consider for example the setting in Section \ref{ssec:known} (i.e., when the parameter under the null hypothesis in known). When $D(\hat{\theta}^{(1)}(n)||\theta^{(0)})\geq\frac{D(\theta^{(0)}||\hat{\theta}^{(1)}(n))}{(M-1)}$, the Chernoff test selects cell $\hat{m}(n)$ and draws the rest $K-1$ cells randomly with equal probability from the remaining $M-1$ cells. When $D(\hat{\theta}^{(1)}(n)||\theta^{(0)})<\frac{D(\theta^{(0)}||\hat{\theta}^{(1)}(n))}{(M-1)}$, all $K$ cells are drawn randomly with equal probability from cells $\{m^{(1)}(n),m^{(2)}(n),\ldots, m^{(M-1)}(n)\}$. The same selection rule is obtained when setting the alternative hypothesis according to the settings in Sections \ref{ssec:uncorrelated}, \ref{ssec:general}. We refer to this policy as the modified Chernoff test.
We present numerical examples to illustrate the performance of the proposed deterministic policy as compared to the randomized Chernoff test, under the setting considered in Section \ref{ssec:general}. 
It can be seen in Figures \ref{fig:si3m3}, \ref{fig:si3m4}, that the proposed deterministic DS algorithm significantly outperforms the randomized Chernoff test.

\begin{figure}[htbp]
	\centering \epsfig{file=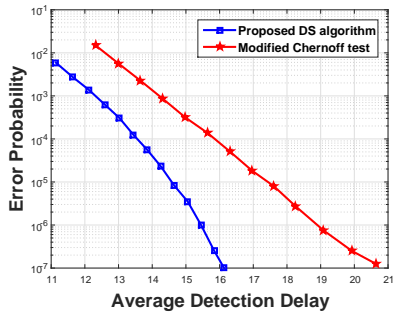,
		width=0.35\textwidth}
	\caption{
		The error probability as a function of the average detection delay under various algorithms: (i) The proposed DS algorithm that uses the MALLR statistics (referred to as the proposed DS algorithm); and (ii) The modified randomized Chernoff test as described in Section \ref{ssec:Chernoff and Albert}. A case of exponential distributions with parameters $\theta ^ {(0)} = 1, \theta ^ {(1)} = 10$ under normal and abnormal states, respectively, where $M = 15$ and $K = 5$. We averaged over $4\cdot 10^7$ Monte Carlo runs.}
	\label{fig:si3m3}
\end{figure}

\begin{figure}[htbp]
	\centering \epsfig{file=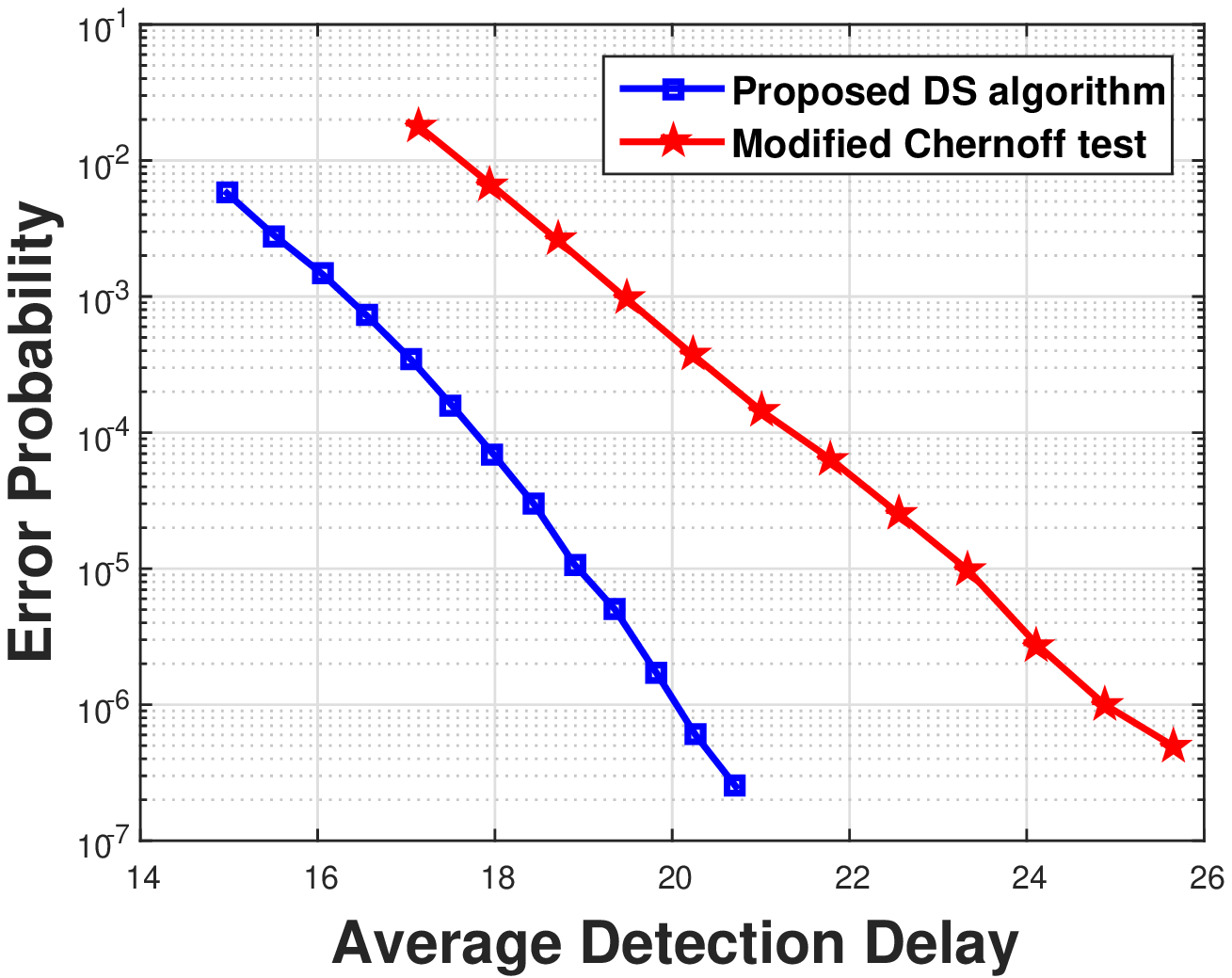,
		width=0.35\textwidth}
	\caption{
		The error probability as a function of the average detection delay under various algorithms: (i) The proposed DS algorithm that uses the MALLR statistics (referred to as the proposed DS algorithm); and (ii) The modified randomized Chernoff test as described in Section \ref{ssec:Chernoff and Albert}. A case of exponential distributions with parameters $\theta ^ {(0)} = 1, \theta ^ {(1)} = 10$ under normal and abnormal states, respectively, where $M = 20$ and $K = 5$. We averaged over $4\cdot 10^7$ Monte Carlo runs.}
	\label{fig:si3m4}
\end{figure}

\section{Empirical Studies}
\label{sec:experiemnts}

In this section, we present additional numerical experiments\footnote{The indifference region in the simulations was set to 
$\displaystyle I= [\frac{\theta^{(0)}+\theta^{(1)}}{2}-10^{-3},\frac{\theta^{(0)}+\theta^{(1)}}{2}+10^{-3}] $. We ran Monte-Carlo experiments for generating the simulation results.}
for demonstrating the performance of the proposed DS algorithm as compared to existing methods. 

\subsection{Comparison between MALLR and LALLR statistics}
\label{ssec:experiments_allr}

We first compare the proposed DS algorithm under the settings of Section \ref{ssec:general}, using the MALLR statistics defined in (\ref{eq:MALLR}) and the LALLR statistics defined in (\ref{eq:MALLR_L}), which is a popular method for performing sequential composite hypothesis testing, first introduced by Robbins and Siegmund in \cite{robbins1972class} (variations can be found in \cite{robbins1974expected, Pavlov_1990_Sequential, Tartakovsky_2002_Efficient}).
It can be seen that the proposed DS algorithm using the MALLR statistics adopts a variation of the LALLR statistics in the design of the stopping rule for anomaly detection over multiple composite hypotheses. However, since both empty cells and the cell that contains the target can be observed by the decision maker, the unconstrained MLEs of the unknown parameters $\theta^{(0)}$ and $\theta^{(1)}$ can be applied in both numerator and denominator (which we referred to as MALLR). 
We next simulate the case of searching for a target over processes that follow Laplace distributions with unknown means, where the observations $y_j$ are drawn from distribution $f\left(y_j|\theta\right)=\frac{1}{2}\exp \left\{|y-\theta|\right\}$.
We note that by using the global MLE we expect for better performance. The simulation results demonstrate the performance gain that we can get in this setting.   
It can be seen in Figure \ref{fig:sim2} that implementing the DS algorithm with MALLR statistics as proposed in Section \ref{ssec:general} significantly outperforms an algorithm that uses the selection rule of DS algorithm with the LALLR statistics as proposed in Section \ref{ssec:uncorrelated}.
It can be seen that the error exponent is significantly better when using the MALLR statistics in the algorithm design. Thus, the performance gain by using the proposed DS algorithm is expected to further increase as the error decreases. 

\begin{figure}[htbp]
\centering \epsfig{file=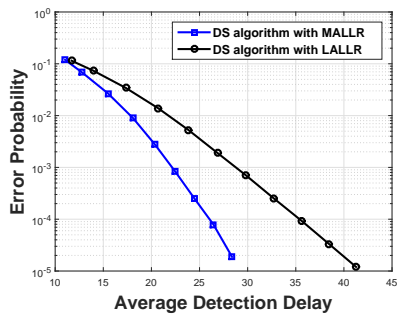,
width=0.35\textwidth}
\caption{The error probability as a function of the average detection delay. Performance comparison between the following algorithms: (i) The proposed DS algorithm that uses the MALLR statistics as described in Section \ref{ssec:general} (referred to as the DS selection rule with MALLR); and (ii) The proposed DS algorithm that uses the LALLR statistics as described in Section \ref{ssec:uncorrelated}. A case of Laplace distributions with parameters $\theta^{(0)}=0$, $\theta^{(1)}=1$, under normal and abnormal states, respectively, with $K=2$, $M=5$. We averaged over $10^6$ Monte Carlo runs.}
\label{fig:sim2}
\end{figure}
	
\subsection{Comparison between MALLR and MGLLR}
\label{ssec:experiments_amllr_gmllr}
In Figure \ref{fig:amllr_gmllr}, we compare the performance of the two proposed statistics suggested in section \ref{ssec:general}. As discussed earlier, using the MALLR statistics allows us to establish asymptotic optimality theoretically, whereas asymptotic optimality remains open when using the MGLLR. However, in practice, we expect that using the MGLLR will perform better since it uses all samples when updating the MLE. It can be seen in Figure \ref{fig:amllr_gmllr} that the DS algorithm using the MGLLR statistics slightly outperforms the DS algorithm using the MALLR.         
\begin{figure}[htbp]
	\centering \epsfig{file=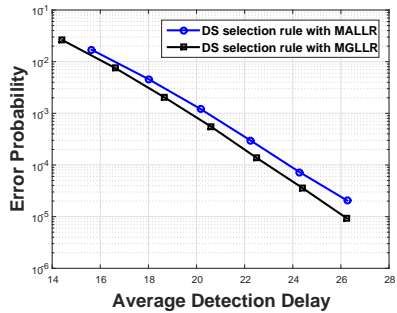,
		width=0.35\textwidth}
	\caption{The error probability as a function of the average detection delay. Performance comparison between the following algorithms: (i) The proposed DS algorithm that uses the MALLR statistics as described in Section \ref{ssec:general} (referred to as the DS selection rule with MALLR); and (ii) The proposed DS algorithm that uses the MGLLR statistics. A case of Laplace distributions with parameters $\theta^{(0)}=0$, $\theta^{(1)}=1$, under normal and abnormal states, respectively, with $K=2$, $M=5$. We averaged over $10^6$ Monte Carlo runs.}
	\label{fig:amllr_gmllr}
\end{figure}

\subsection{Network Traffic Analysis}
\label{ssec:experiments_network_traffic}
Finally, we demonstrate the performance of the DS algorithm using the MALLR statistics in intrusion detection applications, by detecting statistical deviations in network traffic. We examine anomaly detection in packet size statistics, which has been mostly investigated using open loop strategies for detecting malicious activity. We use the model in \cite{Thatte_2011_Parametric} that proposed a sample entropy for packet-size modeling and demonstrated strong performance in detecting anomalous data using the GLR statistics in the sequential detection test. Specifically,
for a given interval, let $\mathcal{S}$ be the set of packet size values that have arrived in this interval, and let $q^{(i)}$ be the proportion of number of packets of size $i$ to the total number of packets that have arrived in that interval. The sample entropy $y$ is thus computed as
$y=-\sum\limits_{i\in\mathcal{S}} q^{(i)}\log q^{(i)}$.
The sample entropy is modeled by Gaussian distribution and given by:
\begin{center}
$p\left( y|\mu, \sigma \right) = {1 \over {\sqrt {2\pi {\sigma ^2}} }}\exp \left[ { - {1 \over {2{\sigma ^2}}}{{\left( {y - \mu } \right)}^2}} \right]$,
\end{center}
where ${\theta ^{\left( 0 \right)}} = \left( {{\mu _0}, {\sigma _0}} \right)$, and ${\theta ^{\left( 1 \right)}} = \left( {{\mu _1}, {\sigma _1}} \right)$, under normal state, or abnormal state, respectively.
We simulated a network with $M$ flows of data, in which a single flow is abnormal. We used the DARPA intrusion detection data set \cite{darpa_dataset}, which contains 5-million labeled network connections, for generating the normal and abnormal flows. When testing the algorithms, the sample entropy has been learned online from the data. We implemented both the proposed DS algorithm, and the entropy-based algorithm with the GLR statistics that has been proposed in \cite{Thatte_2011_Parametric}.
We set the thresholds so that both algorithms satisfy error probability ${10^{ - 4}}$. It can be seen in Figure \ref{fig:realdata} that the DS algorithm achieves strong performance and significantly outperforms the entropy-based algorithm with the GLR statistics.
\begin{figure}[htbp]
\centering \epsfig{file=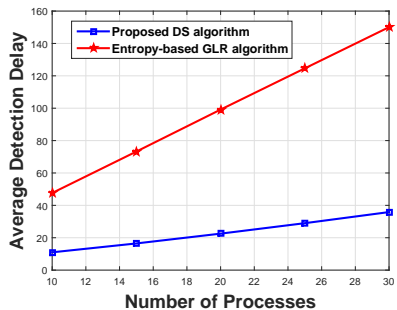,
width=0.35\textwidth}
\caption{The average detection delay as a function of the number of processes $M$ using the DARPA intrusion detection dataset. Performance comparison between the following algorithms: (i) The proposed DS algorithm that uses the MALLR statistics as described in Section \ref{ssec:general} (referred to as the proposed DS algorithm); and (ii) a policy that applies an open loop selection rule when probing cells and uses the GLR statistics for the packet size modeling in the stopping rule as proposed in \cite{Thatte_2011_Parametric} (referred to as entropy-based GLR algorithm). We averaged over $10^5$ Monte Carlo runs.}
\label{fig:realdata}
\end{figure}

\section{Conclusion}
We considered the problem of searching for anomalies among $M$ processes (i.e., cells). The observations follow a common distribution with an unknown parameter, belonging to disjoint parameter spaces depending on whether the target is absent or present. The decision maker is allowed to probe a subset of the cells at a time and the objective is a sequential search strategy that minimizes the expected detection time subject to an error probability constraint. We have developed a deterministic search algorithm to solve the problem that enjoys the following properties.
First, when no additional side information on the process states is known, the proposed algorithm was shown to be asymptotically optimal. Second, when the parameter value under the null hypothesis is known and equal for all normal processes, asymptotic optimality was shown as well, with better detection time determined by the true null state. Third, when the parameter value under the null hypothesis is unknown, but is known to be equal for all normal processes, consistency was shown in terms of achieving error probability that decays to zero with the detection delay. Finally, an explicit upper bound on the error probability under the proposed algorithm was established under the finite sample regime.
Extensive experiments have demonstrated the efficiency of the algorithm over existing methods. \\

\section{Acknowledgment}
We would like to thank the anonymous reviewers for comments that significantly improved the technical results and presentation of this paper.
	
\section{Appendix}
\label{app}
For purposes of presentation, we start by proving Theorem \ref{th:DS_optimality_policy1}. Then, we focus on the key steps for extending the results to the other models presented in Section \ref{sec:DS}. 
\subsection{Proof of Theorem~\ref{th:DS_optimality_policy1}}
\label{app:proof_policy1}

Without loss of generality we prove the theorem when hypothesis $m$ is true. For simplifying the presentation, we start with proving the theorem when the parameter space is finite, so that $\theta^{(0)}, \theta^{(1)}$ can take a finite number of values (but the measurements can still be continuous). We will then extend the proof for continuous parameter space under mild regularity conditions. 
The proof is derived using the adjusted LALLR statistics defined in (\ref{eq:ALLR_Nomral}), i.e.,
$S_m^{(0)}(n)=\sum\limits_{t = 1}^n {{{\bf{1}}_m}} (t)\log {{f({y_m}(t)|{{\hat \theta }_m}(t-1))} \over {f({y_m}(t)|{{ \theta }^{(0)}})}}$. \\

\emph{Step 1: Bounding the error probability:}\vspace{0.1cm}\\
\noindent
We first prove the upper bound on the error probability for all $c$. Specifically, we show below that the error probability is upper bounded by:
\beq
\label{eq:Pe_bound_policy1}
P_e=\sum_{m=1}^{M}\pi_m \alpha_m\leq (M-1)c \;.
\eeq
Let
\begin{center}
$\alpha_{m,j}=\mathbf{P}_m(\delta=j)$
\end{center}
for all $j\neq m$.
Thus,
\begin{center}
$\displaystyle\alpha_m=\sum_{j\neq m}\alpha_{m,j}$.
\end{center}
Therefore, we need to show that $\alpha_{m,j}\leq c$ for proving (\ref{eq:Pe_bound_policy1}). Note that $\alpha_{m,j}$ can be rewritten as follows:
\beq
\label{eq:1}
\bea{l}
\displaystyle\alpha_{m,j}=\mathbf{P}_m\left(\delta=j\right) \vspace{0.1cm} \\\hspace{1.5cm}
=\displaystyle\mathbf{P}_m\left(S_j^{(0)}(\tau)\geq-\log c \mbox{\; for some $\tau\geq 1$}\right)\vspace{0.1cm} \\\hspace{1.5cm}
\leq\displaystyle\mathbf{P}_m\left(Z\left(N_j(\tau)\right)\geq\frac{1}{c} \mbox{\; for some $N_j(\tau)\geq 1$}\right),
\ena
\eeq
where
\beq
\displaystyle Z\left(N_j(\tau)\right)\triangleq e^{ S_j^{(0)}(\tau)}=\prod_{i=1}^{N_j(\tau)}\frac{f(y_j(r_i)|\hat{\theta}_j{(r_i)})}{f(y_j(r_i)|\theta^{(0)})},
\eeq
and $r_1, ..., r_{N_j(\tau)}$ are the time indices in which observations are taken from cell $j$. Next, note that $Z\left(N_j(\tau)\right)$ is a nonnegative martingale,
\beq
\bea{l}
\displaystyle \E_{\theta^{(0)}}\left[Z\left(N_j(\tau)\right)|\left\{y_j(r_i)\right\}_{i=1}^{N_j(\tau)-1}\right] \vspace{0.1cm}\\
\displaystyle =Z\left(N_j(\tau)-1\right)\E_{\theta^{(0)}}\left[\frac{f(y_j(r_{N_j(\tau)})|\hat{\theta}_j{(r_{N_j(\tau)})})}{f(y_j(r_{N_j(\tau)})|\theta^{(0)})}\right]
\vspace{0.1cm}\\\hspace{4cm}
\displaystyle
=Z\left(N_j(\tau)-1\right).
\ena
\eeq
Therefore, applying Lemma 1 in \cite{robbins1972class} for nonnegative martingales yields:
\beq
\bea{l}
\displaystyle\displaystyle\mathbf{P}_m\left(Z\left(N_j(\tau)\right)\geq\frac{1}{c} \mbox{\; for some $N_j(\tau)\geq 1$}\right)\vspace{0.1cm} \\\hspace{5cm}
\leq c\E_{\theta^{(0)}}\left[Z\left(1\right)\right].
\ena
\eeq
Finally, since $\E_{\theta^{(0)}}\left[Z\left(1\right)\right]=1$, we have $\alpha_{m,j}\leq c$, which completes Statement $1$ of the theorem.
\vspace{0.1cm}
	
Next, we define the following major event:
\begin{definition}
$\tau_{ML}$ is the smallest integer such that $\hat{\theta}_m(n)=\theta^{(1)}$, and $\hat{\theta}_j(n)=\theta^{(0)}$ for all $j\neq m$ for all $n>\tau_{ML}$, when $H_m$ is the true hypothesis. \vspace{0.1cm}
\end{definition}
	
\begin{remark}
Note that for all $n>\tau_{ML}$ only the exploitation phase is implemented. As a result, the time spent during the round-robin exploration phase is upper bounded by $\tau_{ML}$. In the next step of the proof we show that $\tau_{ML}$ is bounded, which also yields Statement 3 of the theorem. It should be noted that $\tau_{ML}$ is not a stopping time. The decision maker does not know whether it has arrived. However, it is used to upper bound the actual stopping time under the DS algorithm.
\vspace{0.1cm}
\end{remark}

\begin{remark}
For evaluating the detection time under the DS algorithm, we analyze the case where the \emph{DS algorithm is implemented indefinitely}. When we say that the DS algorithm is implemented indefinitely we mean that we probe the cells as described by the DS algorithm, while disregarding the stopping rule. This analysis enables us to upper bound the actual detection time when the stopping rule is applied.
\vspace{0.1cm}
\end{remark}

\emph{Step 2: Bounding $\tau_{ML}$:}\vspace{0.1cm}
\begin{lemma}
\label{lemma:tau_ML}
Assume that the DS algorithm is implemented indefinitely. Then, there exist $C>0$ and $\gamma>0$ such that
\beq
\label{eq:lemma:T_ML}
\displaystyle\mathbf{P}_m\left(\tau_{ML}>n\right)\leq Ce^{-\gamma n}\;. \vspace{0.4cm}
\eeq
\end{lemma}
%
\begin{proof}
Note that event $\tau_{ML}>n$ implies one of the following events:
(i) There exists a time instant $t>n$ at the round-robin exploration phase, in which $\hat{\theta}_m(t)\neq\theta^{(1)}$, or $\hat{\theta}_j(t)\neq\theta^{(0)}$ for some $j\neq m$. When such time $t$ occurs we say that $E_1(t)$ occurs.
(ii) At the beginning of an exploitation phase (say at time $n'$) $\hat{\theta}_m(n')=\theta^{(1)}$, and $\hat{\theta}_j(n')=\theta^{(0)}$ for all $j\neq m$. However, there exists a time instant $t>n$ during the exploitation phase, in which $\hat{\theta}_m(t)\neq\theta^{(1)}$. When such time $t$ occurs we say that $E_2(t)$ occurs.
		
We can rewrite (\ref{eq:lemma:T_ML}) as follows:
\beq
\bea{l}
\label{eq:Pr_E1_E2}
\displaystyle\mathbf{P}_m\left(\tau_{ML}>n\right)\leq\mathbf{P}_m\left(E_1(t) \mbox{\;occurs for some $t\geq n$}\right) \vspace{0.1cm}\\\hspace{3cm}
\displaystyle+\mathbf{P}_m\left(E_2(t) \mbox{\;occurs for some $t\geq n$}\right) \vspace{0.1cm}\\\hspace{1cm}
\displaystyle\leq\sum_{t=n}^{\infty}\mathbf{P}_m\left(E_1(t) \mbox{\;occurs}\right) 
+\sum_{t=n}^{\infty}\mathbf{P}_m\left(E_2(t) \mbox{\;occurs}\right) \;. \vspace{0.1cm}
\ena
\eeq

Next, we upper bound the first term on the RHS of (\ref{eq:Pr_E1_E2}). It suffices to show that there exist $C>0$ and $\gamma>0$ such that $\mathbf{P}_m\left(E_1(n) \mbox{\;occurs}\right)<Ce^{-\gamma n}$. Let $N_{RR}(n)$ be the total number of time instants spent during the round-robin exploration phase up to time $n$, and fix $0<r<1$. Then, $\mathbf{P}_m\left(E_1(n) \mbox{\;occurs}\right)$ can be rewritten as follows:
\beq
\bea{l}
\label{eq:Pr_E1}
\displaystyle\mathbf{P}_m\left(E_1(n) \mbox{\;occurs}\right)=\mathbf{P}_m\left(E_1(n) \mbox{\;occurs}, N_{RR}(n)\geq rn\right)\vspace{0.1cm}\\\hspace{3cm}
+\mathbf{P}_m\left(E_1(n) \mbox{\;occurs}, N_{RR}(n)<rn\right).
\ena
\eeq
We first upper bound the first term on the RHS of (\ref{eq:Pr_E1}). Since that more than $rn$ observations were taken in a round-robin manner, then at least $rn/M$ observations were taken from each cell. Thus,
\beq
\bea{l}
\label{eq:Pr_E1_first}
\displaystyle\mathbf{P}_m\left(E_1(n) \mbox{\;occurs}, N_{RR}(n)\geq rn\right)\vspace{0.1cm}\\\hspace{1cm}
\displaystyle\leq\mathbf{P}_m\left(\hat{\theta}_m(n)\neq\theta^{(1)}, N_m(n)\geq rn/M\right) \vspace{0.1cm}\\\hspace{1.5cm}
\displaystyle+\sum_{j\neq m}\mathbf{P}_m\left(\hat{\theta}_j(n)\neq\theta^{(0)}, N_j(n)\geq rn/M\right).
\ena
\eeq
Next, we show that the first term on the RHS of (\ref{eq:Pr_E1_first}) decreases exponentially with $n$. Let $\left(y_m(r_1), ..., y_m(r_{N_m(n)})\right)$
be the vector of all $N_m(n)$ observations (indicated by times $r_1, ..., r_{N_m(n)}$) collected from cell $m$ up to time $n$, and let $\tilde{\theta}_m(n')=\hat{\theta}_m(n)$ denotes the MLE based on $N_m(n)=n'$ observations collected from cell $m$ up to time $n$.
We can upper bound $\mathbf{P}_m\left(\hat{\theta}_m(n)\neq\theta^{(1)}, N_m(n)\geq rn/M\right)$ by:
\beq
\bea{l}
\label{eq:Pr_E1_first_first}
\displaystyle\mathbf{P}_m\left(\hat{\theta}_m(n)\neq\theta^{(1)}, N_m(n)\geq rn/M\right)\vspace{0.1cm}\\
\displaystyle\leq\sum_{q=\lceil rn/M\rceil}^{\infty}\mathbf{P}_m\left(\tilde{\theta}_m(q)\neq\theta^{(1)}\right).
\ena
\eeq
Then, by the definition of the MLE (\ref{eq:unconstrained_MLE}), the event $\tilde{\theta}_m(n)\neq\theta^{(1)}$ implies:
\beq
\displaystyle\sum_{i=1}^{n}\ell_{\theta^{(1)}, \tilde{\theta}_m(n)}(i)<0,
\eeq
for some $\tilde{\theta}_m(n)\neq\theta^{(1)}$, where
\begin{center}
$\displaystyle\ell_{\theta^{(1)}, \tilde{\theta}_m(n)}(i)
\triangleq\log\frac{f\left(y_m(i)|\theta^{(1)}\right)}{f\left(y_m(i)|\tilde{\theta}_m(n)\right)}$.
\end{center}
Note that we only refer to the number of observations irrespective of the probing times due to i.i.d. property.
Hence, it remains to show that $\mathbf{P}_m\left(\sum_{i=1}^{n}\ell_{\theta^{(1)}, \tilde{\theta}_m(n)}(i)<0\right)$ decreases exponentially with $n$ for each $\tilde{\theta}_m(n)\neq\theta^{(1)}$. Applying the Chernoff bound and using the i.i.d. property yields:
\beq
\bea{l}
\displaystyle\mathbf{P}_m\left(\sum_{i=1}^{n}\ell_{\theta^{(1)}, \tilde{\theta}_m(n)}(i)<0\right)\vspace{0.1cm}\\\hspace{3cm}
\displaystyle\leq\displaystyle\left[\mathbf{E}_m\left(e^{s(-\ell_{\theta^{(1)}, \tilde{\theta}_m(n)}(i))}\right)\right]^{n}.
\ena
\eeq
Note that a moment generating function (MGF) is equal to one at $s=0$. Furthermore, since $\mathbf{E}_m(-\ell_{\theta^{(1)}, \tilde{\theta}_m(n)}(i))=-D\left(\theta^{(1)}||\tilde{\theta}_m(n)\right)<0$ is strictly negative, differentiating the MGFs of $-\ell_{\theta^{(1)}, \tilde{\theta}_m(n)}(i)$ with respect to $s$ yields a strictly negative derivative at $s=0$. Hence, there exist $s>0$ and $\gamma'>0$ such that $\mathbf{E}_m\left(e^{s(-\ell_{\theta^{(1)}, \tilde{\theta}_m(n)}(i))}\right)$ is strictly less than $e^{-\gamma'}<1$, which yields the desired exponential decay. A similar argument applies for showing that the second term on the RHS of (\ref{eq:Pr_E1_first}) decreases exponentially with $n$.
		
Next, we upper bound the second term on the RHS of (\ref{eq:Pr_E1}). Let $N_{XT}$, $N_{XT,j}$ be the total number of time instants spent during the exploitation phase up to time $n$ at all cells and cell $j$, respectively. Since $N_{RR}(n)<rn$, then $N_{XT}\geq (1-r)n$. Fix $0<r_2<1$. We can rewrite the second term on the RHS of (\ref{eq:Pr_E1}) as follows:
\beq
\bea{l}
\label{eq:Pr_E1_second}
\displaystyle\mathbf{P}_m\left(E_1(n) \mbox{\;occurs}, N_{RR}(n)<rn\right)\vspace{0.1cm}\\\hspace{0cm}
\displaystyle\leq\mathbf{P}_m\left(E_1(n) \mbox{\;occurs}, N_{XT}(n)\geq(1-r)n, \right.\vspace{0.1cm}\\\hspace{4cm}\displaystyle\left.
N_{XT,m}(n)\geq r_2(1-r)n\right)\vspace{0.1cm}\\\hspace{0cm}
\displaystyle+\mathbf{P}_m\left(E_1(n) \mbox{\;occurs}, N_{XT}(n)\geq(1-r)n, \right.\vspace{0.1cm}\\\hspace{4cm}\displaystyle\left.
N_{XT,m}(n)<r_2(1-r)n\right).
\ena
\eeq
We first upper bound the first term on the RHS of (\ref{eq:Pr_E1_second}). Note that $E_1(n)$ occurs implies that there exists an exploitation time $t$ before time $n$, in which cell $m$ has been probed, its MLE was computed based on more than $r_2(1-r)n$ observations, and error event was occurred, $\hat{\theta}_m(t)\neq\theta^{(1)}$ (so that the algorithm moved back to exploration phase and $E_1(n)$ occurred).
Therefore, we can write:
\beq
\bea{l}
\label{eq:Pr_E1_second_first}
\displaystyle\mathbf{P}_m\left(E_1(n) \mbox{\;occurs}, N_{XT}(n)\geq(1-r)n, \right.\vspace{0.1cm}\\\hspace{4cm}\displaystyle\left.
N_{XT,m}(n)\geq r_2(1-r)n\right)\vspace{0.1cm}\\\hspace{0cm}
\displaystyle\leq\sum_{t=\lceil r_2(1-r)n\rceil}^{\infty}\mathbf{P}_m\left(\hat{\theta}_m(t)\neq\theta^{(1)}, N_m(t)\geq r_2(1-r)n\right).
\ena
\eeq
By a similar argument as we developed when proving (\ref{eq:Pr_E1_first}), each of the terms in the summation decreases exponentially with $n$, implying exponentially decreasing of the first term on the RHS of (\ref{eq:Pr_E1_second}).
		
Next, we upper bound the second term on the RHS of (\ref{eq:Pr_E1_second}). Since less than $r_2(1-r)n$ observations were taken from cell $m$ during exploitation and the total number of observations during exploitation is more than $(1-r)n$, then there exists cell $j\neq m$ that have been observed more than $(1-r_2)(1-r)n/(M-1)$ times during exploitation phase. This implies that there exists an exploitation time $t$ before time $n$, in which cell $j$ has been probed, its MLE was computed based on more than $(1-r_2)(1-r)n/(M-1)=qn$ observations, where $0<q\triangleq(1-r_2)(1-r)/(M-1)<1$, and $\hat{\theta}_j(t)\neq\theta^{(0)}$
Therefore, we can write:
\beq
\bea{l}
\label{eq:Pr_E1_second_second}
\displaystyle\mathbf{P}_m\left(E_1(n) \mbox{\;occurs}, N_{XT}(n)\geq(1-r)n, \right.\vspace{0.1cm}\\\hspace{4cm}\displaystyle\left.
N_{XT,m}(n)< r_2(1-r)n\right)\vspace{0.1cm}\\\hspace{0cm}
\displaystyle\leq\sum_{t=\lceil qn\rceil}^{\infty}\mathbf{P}_m\left(\hat{\theta}_j(t)\neq\theta^{(0)}, N_j(t)\geq qn\right).
\ena
\eeq
By a similar argument as we developed when proving (\ref{eq:Pr_E1_first}), each of the terms in the summation decreases exponentially with $n$, implying exponentially decreasing of the second term on the RHS of (\ref{eq:Pr_E1_second}). Therefore, we have shown exponentially decreasing of the first term on the RHS of (\ref{eq:Pr_E1_E2}).\vspace{0.1cm}
		
It remains to show an exponentially decreasing of the second term on the RHS of (\ref{eq:Pr_E1_E2}). It suffices to show that there exist $C>0$ and $\gamma>0$ such that $\mathbf{P}_m\left(E_2(n) \mbox{\;occurs}\right)<Ce^{-\gamma n}$. Fix $0<r<1$. We can rewrite $\mathbf{P}_m\left(E_2(n) \mbox{\;occurs}\right)$ as follows:
\beq
\bea{l}
\label{eq:Pr_E2}
\displaystyle\mathbf{P}_m\left(E_2(n) \mbox{\;occurs}\right)\leq\mathbf{P}_m\left(E_2(n) \mbox{\;occurs}, N_m(n)\geq rn\right)\vspace{0.1cm}\\\hspace{3cm}
+\mathbf{P}_m\left(E_2(n) \mbox{\;occurs}, N_m(n)<rn\right).
\ena
\eeq
We first upper bound the first term on the RHS of (\ref{eq:Pr_E2}). Since $E_2(n)$ occurs and more than $rn$ observations were taken from cell $m$ we have:
\beq
\bea{l}
\label{eq:Pr_E2_first}
\displaystyle\mathbf{P}_m\left(E_2(n) \mbox{\;occurs}, N_m(n)\geq rn\right)\vspace{0.1cm}\\\hspace{1cm}
\displaystyle\leq\mathbf{P}_m\left(\hat{\theta}_m(n)\neq\theta^{(1)}, N_m(n)\geq rn\right).
\ena
\eeq
By a similar argument as we developed when proving (\ref{eq:Pr_E1_first}), the RHS of (\ref{eq:Pr_E2_first}) decreases exponentially with $n$.
		
Next, we upper bound the second term on the RHS of (\ref{eq:Pr_E2}). Since $N_m(n)<rn$ then at least $(1-r)n$ observations were taken from other cells. Let $\tilde{N}_{RR}(n)$ be the total number of observations collected from all cells excepts cell $m$ during the round-robin exploration phase up to time $n$, and fix $0<r_2<1$. Then, the second term on the RHS of (\ref{eq:Pr_E2}) can be rewritten as follows:
\beq
\bea{l}
\label{eq:Pr_E2_second_N}
\displaystyle\mathbf{P}_m\left(E_2(n) \mbox{\;occurs}, N_m(n)<rn\right)\vspace{0.1cm}\\
\displaystyle=\mathbf{P}_m\left(E_2(n) \mbox{\;occurs}, N_m(n)<rn, \tilde{N}_{RR}(n)\geq r_2(1-r)n\right)\vspace{0.1cm}\\\hspace{0cm}
\displaystyle+\mathbf{P}_m\left(E_2(n) \mbox{\;occurs}, N_m(n)<rn, \tilde{N}_{RR}(n)<r_2(1-r)n\right).
\ena
\eeq
Next, we upper bound the first term on the RHS of (\ref{eq:Pr_E2_second_N}). Since more than $r_2(1-r)n$ observations were taken from all cells excepts cell $m$ during round-robin exploration, then at least $r_2(1-r)n/(M-1)$ observations were taken from each cell $j\neq m$ (and the same number of observations must have been taken from cell $m$ as well during round-robin exploration). Then, at time $n$ during exploration phase, its MLE was computed based on more than $r_2(1-r)n/(M-1)$ observations, and error event was occurred, $\hat{\theta}_m(t)\neq\theta^{(1)}$ (so that $E_2(n)$ occurred). Then,
\beq
\bea{l}
\label{eq:Pr_E2_second_N_first}
\displaystyle\mathbf{P}_m\left(E_2(n) \mbox{\;occurs}, N_m(n)<rn, \tilde{N}_{RR}(n)\geq r_2(1-r)n\right)\vspace{0.1cm}\\
\displaystyle\leq\mathbf{P}_m\left(\hat{\theta}_m(n)\neq\theta^{(1)}, N_m(n)\geq r_2(1-r)n/(M-1)\right).
\ena
\eeq
By a similar argument as we developed when proving (\ref{eq:Pr_E1_first}), the RHS of (\ref{eq:Pr_E2_second_N_first}) decreases exponentially with $n$.
		
Next, we upper bound the second term on the RHS of (\ref{eq:Pr_E2_second_N}). Since less than $r_2(1-r)n$ observations were taken from all cells excepts cell $m$ during round-robin exploration, then there exists cell $j\neq m$ in which more than $(1-r_2)(1-r)n/(M-1)$ observations were taken from it during exploitation phase. By subtracting all time instants in which the test might switch between exploration to exploitation phases, at least $(1-r_2)(1-r)n/(M-1)-r_2(1-r)n$ observations were taken during exploitation, where $\hat{\theta}_j(t)\neq\theta^{(0)}$.
 We can choose small $r_2$ (e.g., $r_2=1/(3(M-1))$) so that $(1-r_2)(1-r)n/(M-1)-r_2(1-r)n=qn$ for $0<q<1$. Thus,
\beq
\bea{l}
\label{eq:Pr_E2_second_N_second}
\displaystyle\mathbf{P}_m\left(E_2(n) \mbox{\;occurs}, N_m(n)<rn, \tilde{N}_{RR}(n)< r_2(1-r)n\right)\vspace{0.1cm}\\
\displaystyle\leq\mathbf{P}_m\left(\hat{\theta}_j(n)\neq\theta^{(0)}, N_j(n)\geq qn\right).
\ena
\eeq
By a similar argument as we developed when proving (\ref{eq:Pr_E1_first}), the RHS of (\ref{eq:Pr_E2_second_N_first}) decreases exponentially with $n$.
Hence, (\ref{eq:lemma:T_ML}) follows.
\end{proof}
	
Note that the total time spent during the round-robin exploration phase is upper bounded by $\tau_{ML}$. Hence, Statement 3 in Theorem \ref{th:DS_optimality_policy1} follows.\vspace{0.1cm}
	
\emph{Step 3: Bounding the detection time:}\vspace{0.1cm}

\begin{definition}
Assume that the DS algorithm is implemented indefinitely. Then, $\tau_U$ denotes the first time that $S_m^{(0)}(n)\geq -\log(c)$ for $n>\tau_{ML}$:
\beq
\displaystyle\tau_U\triangleq\inf\left\{n>\tau_{ML}\;:\;S_m^{(0)}(n)\geq-\log c\right\},
\eeq
and $n_U\triangleq\tau_U-\tau_{ML}$ denotes the total amount of time between $\tau_{ML}$ and $\tau_U$.
\vspace{0.1cm}
\end{definition}

It should be noted that the actual detection time $\tau$ under DS algorithm (when the stopping rule is applied) is upper bounded by $\tau_U$.
In the next lemma we show that $n_U$ cannot be significantly larger than $-(1+\epsilon)\log c/D\left(\theta^{(1)}||\theta^{(0)}\right)$ with high probability. \vspace{0.1cm}

\begin{lemma}
\label{lemma:nU}
Assume that the DS algorithm is implemented indefinitely and $H_m$ is true.
Then, for every fixed $\epsilon>0$ there exist $C>0$ and $\gamma>0$ such that
\beq
\bea{l}
\label{eq:lemma:nU}
\mathbf{P}_m\left(n_U>n\right)\leq C e^{-\gamma n} \vspace{0.1cm}\\\hspace{3cm}
\forall n>-(1+\epsilon)\log c/D\left(\theta^{(1)}||\theta^{(0)}\right)\;.
\ena
\eeq
\vspace{0.1cm}
\end{lemma}
%

\begin{proof}
We define
\beq
\tilde{\ell}_m(t)\triangleq\ell_m(t)-D(\theta^{(1)}||\theta^{(0)}),
\eeq
where the MALLR $\ell_m(t)$ at time $n\geq t$ is given by:
\beq
\ell_m(t)\triangleq\log {{f({y_m}(t)|{{\hat \theta }_m}(t))} \over {f({y_m}(t)|{{\hat \theta }^{(r)}}(n))}}.
\eeq
Recall that the test statistics is given by
$S_m^{(0)}(n)=\sum_{t=1}^{n}{\bf{1}}_m (t)\ell_m(t)$. Since that for all $t\geq\tau_{ML}$ the DS algorithm collects observations from cell $m$, then ${\bf{1}}_m (t)=1$ for all $t\geq\tau_{ML}$.
Let $\epsilon_1=D(\theta^{(1)}||\theta^{(0)})\epsilon/(1+\epsilon)>0$. Then, we can write
\beq
\bea{l}
\displaystyle\sum_{i=1}^{\tau_{ML}+n}{\bf{1}}_m (i)\ell_m(i)+\log c \vspace{0.1cm}\\
\displaystyle=\sum_{i=1}^{\tau_{ML}}{\bf{1}}_m (i)\ell_m(i)+\sum_{i=\tau_{ML}+1}^{\tau_{ML}+n}\ell_m(i)+\log c \vspace{0.1cm}\\
=\displaystyle\sum_{i=1}^{\tau_{ML}}{\bf{1}}_m (i)\ell_m(i)+\sum_{i=\tau_{ML}+1}^{\tau_{ML}+n}\tilde{\ell}_m(i)
+n D(\theta^{(1)}||\theta^{(0)})+\log c
\vspace{0.1cm}\\ 
\geq\displaystyle\sum_{i=1}^{\tau_{ML}}{\bf{1}}_m (i)\ell_m(i)+\sum_{i=\tau_{ML}+1}^{\tau_{ML}+n}\tilde{\ell}_m(i)+n\epsilon_1
\;,
\ena
\eeq
for all $n>-(1+\epsilon)\log c/D(\theta^{(1)}||\theta^{(0)})$. \\
As a result,
\beq
\bea{l}
\displaystyle\sum_{i=1}^{\tau_{ML}+n}{\bf{1}}_m (i)\ell_m(i)\leq-\log c \;.
\ena
\eeq
implies
\beq
\bea{l}
\displaystyle\sum_{i=1}^{\tau_{ML}}{\bf{1}}_m (i)\ell_m(i)+\sum_{i=\tau_{ML}+1}^{\tau_{ML}+n}\tilde{\ell}_m(i)\leq-n\epsilon_1 \;.
\ena
\eeq
Hence, for any $\epsilon>0$ there exists $\epsilon_1>0$ such that
\beq
\label{eq:Pm_nU}
\bea{l}
\displaystyle\mathbf{P}_m\left(n_U>n\right)\vspace{0.1cm}\\
\displaystyle\leq\mathbf{P}_m\left(\sum_{i=1}^{\tau_{ML}+n}{\bf{1}}_m (i)\ell_m(i)\leq-\log c\right)\vspace{0.1cm}\\
\displaystyle\leq\mathbf{P}_m\left(\sum_{i=1}^{\tau_{ML}}{\bf{1}}_m (i)\ell_m(i)+\sum_{i=\tau_{ML}+1}^{\tau_{ML}+n}\tilde{\ell}_m(i)\leq-n\epsilon_1\right)\vspace{0.1cm}\\
\displaystyle\leq\mathbf{P}_m\left(\sum_{i=1}^{\tau_{ML}}{\bf{1}}_m (i)\ell_m(i)\leq-n\epsilon_1/2\right)\vspace{0.1cm}\\
\displaystyle+\mathbf{P}_m\left(\sum_{i=\tau_{ML}+1}^{\tau_{ML}+n}\tilde{\ell}_m(i)\leq-n\epsilon_1/2\right)\vspace{0.1cm}\\
\displaystyle\leq\mathbf{P}_m\left(\sum_{i=1}^{\tau_{ML}}{\bf{1}}_m (i)\ell_m(i)\leq-n\epsilon_1/2, \tau_{ML}>\epsilon_2 n\right)\vspace{0.1cm}\\
\displaystyle+\mathbf{P}_m\left(\sum_{i=1}^{\tau_{ML}}{\bf{1}}_m (i)\ell_m(i)\leq-n\epsilon_1/2, \tau_{ML}\leq\epsilon_2 n\right)\vspace{0.1cm}\\
\displaystyle+\mathbf{P}_m\left(\sum_{i=\tau_{ML}+1}^{\tau_{ML}+n}\tilde{\ell}_m(i)\leq-n\epsilon_1/2\right)\vspace{0.1cm}\\
 \;,
\ena
\eeq
for all $n>-(1+\epsilon)\log c/D(\theta^{(1)}||\theta^{(0)})$, and $0<\epsilon_2<1$.
The first term on the RHS decreases exponentially by Lemma \ref{lemma:tau_ML}. Since $\epsilon_2>0$ can be arbitrarily small, and $\ell_m(i)$ has finite expectation, then the second term decreases exponentially by applying the Chernoff bound. Since $\tilde{\ell}_m(i)$ has zero mean for all $i>\tau_{ML}$, then the third term decrease exponentially by applying the Chernoff bound. Hence, (\ref{eq:lemma:nU}) follows.
\vspace{0.1cm}
\end{proof}

Next, we can upper bound the actual detection time under DS algorithm by combining Lemmas \ref{lemma:tau_ML}, \ref{lemma:nU}:
\beq
\label{eq:delay}
\displaystyle\mathbf{E}_m(\tau)\leq\mathbf{E}_m(\tau_{ML})+\mathbf{E}_m(n_U)\leq -\left(1+o(1)\right)\frac{\log(c)}{D(\theta^{(1)}||\theta^{(0)})}.
\eeq

Next, we can upper bound the Bayes risk under DS algorithm By combining (\ref{eq:delay}) and (\ref{eq:Pe_bound_policy1}):
\beq
\displaystyle R_m(\Gamma)\leq -\left(1+o(1)\right)\frac{c\log(c)}{D(\theta^{(1)}||\theta^{(0)})}.
\eeq

Finally, Combining the upper bound on the Bayes risk with the lower bound on the Bayes risk $R_m(\Gamma)\geq-\left(1+o(1)\right)\frac{c\log(c)}{D(\theta^{(1)}||\theta^{(0)})}$ that was obtained in \cite{cohen2015active} under simple hypotheses completes the proof.
\newcommand*{\QEDA}{\hfill\ensuremath{\blacksquare}}
\QEDA
\vspace{0.1cm}

\subsubsection{Extending the Proof of Theorem~\ref{th:DS_optimality_policy1} for Continuous Parameter Space}
\label{ssec:extending}

Next, we focus on the key steps used for extending the proof of Theorem~\ref{th:DS_optimality_policy1} for continuous parameter space.   We need the following requirement for the consistency of the MLE: For all  $\epsilon>0$, we require that $\mathbf{P}_m\left(|\tilde{\theta}_m(n)-\theta^{(1)}|>\epsilon\right)$ decays only polynomially with $n$. To achieve this, we require that the parameters space $\Theta^{(0)}$, $\Theta^{(1)}$ are open sets. Then, the condition holds for a wide class of distributions, including exponential family distributions (see e.g., \cite{fu1993large}). \vspace{0.1cm}
\\
\noindent
\emph{Step 1: Bounding the error probability:}\vspace{0.1cm}\\
\noindent
Bounding the error probability in Step 1 in Appendix \ref{app:proof_policy1} directly applies to continuous parameter space. \vspace{0.1cm}\\
\noindent
\emph{Step 2: Bounding $\tau_{ML}(\epsilon_3)$:}\vspace{0.1cm}\\\
\noindent
Since the MLEs take continuous values, instead of defining $\tau_{ML}$ as in Appendix \ref{app:proof_policy1}, we define $\tau_{ML}(\epsilon_3)$ for some $\epsilon_3>0$ as the smallest integer such that $|\hat{\theta}_m(n)-\theta^{(1)}|\leq\epsilon_3$, and $|\hat{\theta}_j(n)-\theta^{(0)}|\leq\epsilon_3$ for all $j\neq m$ for all $n>\tau_{ML}(\epsilon_3)$, when $H_m$ is the true hypothesis. We require that the parameters take values in the interior of the parameter spaces. Then, we can choose a sufficiently small $\epsilon_3$ so that for all $n>\tau_{ML}(\epsilon_3)$ only the exploitation phase is implemented. As a result, the time spent in a round-robin exploration phase is upper bounded by $\tau_{ML}(\epsilon_3)$ similarly to upper bounding the round-robin exploration time by $\tau_{ML}$ as in Appendix \ref{app:proof_policy1}.

We then modify Lemma \ref{lemma:tau_ML} so that to show at least polynomial decay of $\mathbf{P}_m\left(\tau_{ML}(\epsilon_3)>n\right)$ for any $\epsilon_3>0$ (since polynomial decay is sufficient to guarantee a finite expected value of $\tau_{ML}(\epsilon_3)$). Proving the modified lemma requires similar steps as in Appendix \ref{app:proof_policy1} with the following modification. Since the MLEs take continuous values, instead of referring to the events in which the MLEs are not equal to the true parameter values $\hat{\theta}_m(n)\neq\theta^{(1)}$, and $\hat{\theta}_j(n)\neq\theta^{(0)}$ for all $j\neq m$ as in Appendix \ref{app:proof_policy1}, we refer to the events in which the MLEs deviate from the true parameter values by $\epsilon_3$. As a result, for bounding $\mathbf{P}_m\left(\tau_{ML}(\epsilon_3)>n\right)$ (as in equation (\ref{eq:Pr_E1_first_first})) we need to require only weak consistency of the MLEs so that $\mathbf{P}_m\left(|\tilde{\theta}_m(n)-\theta^{(1)}|>\epsilon_3\right)$ decays only polynomially with $n$ as mentioned above. \vspace{0.1cm}\\
\noindent
\emph{Step 3: Bounding the detection time:}\vspace{0.1cm}\\\
\noindent
Step 3 follows similar steps as in Appendix \ref{app:proof_policy1} for any $\epsilon_3>0$. Since $\epsilon_3>0$ is arbitrarily small, the theorem follows.

\subsection{Proof of Theorem~\ref{th:DS_optimality_policy1_2}}
\label{app:proof_policy1_2}
We now prove the asymptotic optimality criterion given in (\ref{eq:Frequentist_formulation}). Here we are interested in detection procedures that satisfy the constraint $P_e(\Gamma) \leq \alpha$. The class of such detection algorithms will be denoted by $C(\alpha)$.\\
By applying the same steps as in the proof of Theorem 1, we can show that: $P_e(\Gamma) \leq (M-1)e^{-a} $.
It follows that $a = \log(\frac{M-1}{\alpha})$ implies that $\Gamma \in C(\alpha) $. \\	
Next, the upper bound for $\tau_{ML}$ holds with the same steps as in Lemma \ref{lemma:tau_ML}. We define $\tau_U$ and $n_U$ as in \ref{app:proof_policy1}.
To prove the asymptotic optimality, first note that the lower bound on the detection time is given by:
\beq
\label{lemma:upper_bound}
\displaystyle \inf_{\Gamma \in C(\alpha)} \mathbb{E}_m(\tau|\Gamma) \geq \frac{|\log \alpha|}{D(\theta^{(1)}||\theta^{(0)})}(1+o(1)), \; \;  \alpha \rightarrow 0,
\eeq
which can be derived following the same steps as in \cite{tartakovsky2002efficient}.
We next provide the proof for the upper bound on the detection time.
\begin{lemma}
\label{lemma:lower_bound}
Assume that the DS algorithm is implemented indefinitely. Then, 
\beq
\label{eq:lower_bound}
\displaystyle \mathbb{E}_m(\tau|\Gamma^*) \leq \frac{|\log \alpha|}{D(\theta^{(1)}||\theta^{(0)})}(1+o(1)), \; \;  \alpha \rightarrow 0. 
\eeq \\
\end{lemma} 

\begin{proof}
We define the last exit times $L(\epsilon, \theta)$. For all $\epsilon>0$:
\beq
\label{eq:exit_times}
\displaystyle L(\epsilon, \theta)=sup \{n \geq \tau_{ML} |\frac{S_m^{(0)}(n)}{n}-D(\theta^{(1)}||\theta^{(0)}) | > \epsilon  \}.
\eeq
Under $H_m$, \vspace{0.1cm} \\ 
\vspace{0.1cm} $S_m^{(0)}(n_U-1) \geq (n_U-1)(D(\theta^{(1)}||\theta^{(0)})- \epsilon) $ \\
\vspace{0.1cm} on $\{ n_U > L(\epsilon, \theta)+1 \}$, and \\
\vspace{0.1cm} $S_m^{(0)}(n_U-1) < a $, on $\{ n_U < \infty \}$. \\
\vspace{0.1cm}Therefore, for every $ 0< \epsilon < D(\theta^{(1)}||\theta^{(0)})$,\\
\vspace{0.1cm} $ \displaystyle n_U < (1+\frac{a}{D(\theta^{(1)}||\theta^{(0)})- \epsilon}) \mathbf{1}_{\{n_U>1+L(\epsilon,\theta)\}} \\
\vspace{0.1cm} + [1 + L(\epsilon,\theta)] \mathbf{1}_{ \{n_U<1+L(\epsilon,\theta)\}} \\
\vspace{0.1cm} \displaystyle \leq 1 + L(\epsilon,\theta)+ \frac{a}{D(\theta^{(1)}||\theta^{(0)}) - \epsilon}$. \\
By using Chernoff bound we can show that $E[L(\epsilon,\theta)] < \infty$, and by letting $ \epsilon \rightarrow 0 $ and choosing $a = \log(\frac{M-1}{\alpha})$ we get: \vspace{0.1cm} \\
\vspace{0.1cm} $\displaystyle E_m[n_U] \leq \frac{|\log \alpha|}{D(\theta^{(1)}||\theta^{(0)})}(1+o(1))$ as $\alpha \rightarrow 0  $, \\
and combining with the upper bound of $\tau_{ML}$ derived in Lemma \ref{lemma:tau_ML}, we prove the Lemma.
\end{proof}
Finally, by combining Lemma \ref{lemma:lower_bound}, with (\ref{lemma:upper_bound}) we complete the proof.

\subsection{Proof of Theorem~\ref{th:DS_optimality_policy_uncorrelated}}
\label{app:proof_policy_uncorrelated}
We focus on the key steps used for extending the proof of Theorem \ref{th:DS_optimality_policy1} to the settings in which both parameters under normal and abnormal states $\theta^{(0)},\theta^{(1)}$ are unknown, and no additional side information of the parameter values are given. Without loss of generality we prove the theorem when hypothesis $m$ is true, and the proof is derived using the LALLR statistics defined in (\ref{eq:MALLR_L}), i.e., $S_m^{(0)}(n) = S_{m,LALLR}^{(0)}(n)$. \vspace{0.1cm} \\
\emph{Step 1: Bounding the error probability:}\vspace{0.1cm}\\
We begin by upper bounding the error probability for all $c$.
With the same notation as in Step 1 in Appendix \ref{app:proof_policy1}, we need to show that $\alpha_{m,j}=\mathbf{P}_m(\delta=j) \leq c$. We first notice: \vspace{0.1cm} \\
$\vspace{0.1cm} \hspace{0.0cm} S_j^{(0)}(\tau) = \displaystyle \sum\limits_{t = 1}^\tau {{{\bf{1}}_j}} (t)\log {{f({y_j}(t)|{{\hat \theta }_j}(t-1))} \over {f({y_j}(t)|{{\hat{\theta}_j^{(0)}}}(\tau))}} \\
\vspace{0.1cm} \hspace{0.8cm} = \displaystyle \min_{\varphi \in \Theta^{(0)}} \sum\limits_{t = 1}^\tau {{{\bf{1}}_j}} (t)\log {{f({y_j}(t)|{{\hat \theta }_j}(t-1))} \over {f({y_j}(t)|{{\varphi})}} } \\
\vspace{0.1cm} \hspace{0.8cm} \leq \displaystyle \sum\limits_{t = 1}^\tau {{{\bf{1}}_j}} (t)\log {{f({y_j}(t)|{{\hat \theta }_j}(t-1))} \over {f({y_j}(t)|{{\theta^{(0)}})}} }$. \\
Hence, we have: \vspace{0.1cm} \\
$\vspace{0.1cm} \hspace{0.0cm} \alpha_{m,j} = \mathbf{P}_m(\delta=j) \\
 \vspace{0.1cm} \hspace{0.0cm}  = \displaystyle\mathbf{P}_m\left(S_j^{(0)}(\tau)\geq-\log c \mbox{\; for some $\tau\geq 1$}\right)\\
\vspace{0.1cm} \hspace{0.0cm} \leq \displaystyle\mathbf{P}_m \bigg( \sum\limits_{t = 1}^\tau {{{\bf{1}}_j}} (t)\log {{f({y_j}(t)|{{\hat \theta }_j}(t-1))} \over {f({y_j}(t)|{{\theta^{(0)}})}} } \geq-\log c \\
\vspace{0.1cm} \hspace{5.5cm} \mbox{\; for some $\tau\geq 1$} \bigg)$. \\
Next, we can use similar steps as in Appendix \ref{app:proof_policy1}, starting at (\ref{eq:1}) onwards, to prove that $\alpha_{m,j} \leq c $, which implies that the error probability is upper bounded by $(M-1)c$ for all $c$. Thus, Statement $1$ in Theorem \ref{th:DS_optimality_policy_uncorrelated} follows. \vspace{0.1cm} \\
\emph{Step 2: Bounding $\tau_{ML}$:} \vspace{0.1cm}\\   
Upper bounding $\tau_{ML}$ (see (\ref{eq:lemma:T_ML})) follows the same steps as in Lemma \ref{lemma:tau_ML}. Hence, Statement $3$ in Theorem \ref{th:DS_optimality_policy_uncorrelated} follows. \vspace{0.1cm} \\
\emph{Step 3: Bounding the detection time:} \vspace{0.1cm}\\ 
We define $\tau_U$ and $n_U$ similarly as in step 3 in Appendix \ref{app:proof_policy1}:
\begin{definition}
Assume that the DS algorithm is implemented indefinitely. Then, $\tau_U$ denotes the first time that $S_m^{(0)}(n)\geq -\log(c)$ for $n>\tau_{ML}$:
\beq
\displaystyle\tau_U\triangleq\inf\left\{n>\tau_{ML}\;:\;S_m^{(0)}(n)\geq-\log c\right\},
\eeq
and $n_U\triangleq\tau_U-\tau_{ML}$ denotes the total amount of time between $\tau_{ML}$ and $\tau_U$.
\vspace{0.1cm}
\end{definition}
We also define $\displaystyle \ell_m^{(\theta,\varphi)}(t) = \log {{f({y_m}(t)|{\theta})} \over {f({y_m}(t)|{{\varphi})}}} $. \vspace{0.1cm} \\
Note that $S_m^{(0)}(n) =  \sum\limits_{t = 1}^n {{{\bf{1}}_j}}(t) \ell_m^{(\theta^{(1)},\hat{\theta}_m^{(0)}(n))}(t)$ for all $n>\tau_{ML}$.
Define $\tau_U(\varphi)$ to be the first time that 
$\displaystyle \sum\limits_{t = 1}^n {{{\bf{1}}_m}} (t) \ell_m^{(\theta^{(1)},\varphi)}(t) \geq -\log c$ for $n>\tau_{ML}$, and define $n_U(\varphi)=\tau_U(\varphi)-\tau_{ML}$. 
Clearly, $n_U \leq \displaystyle n_U(\varphi)$ for each $\varphi \in \Theta^{(0)}$. We now bound $n_U(\varphi)$ for each $\varphi \in \Theta^{(0)}$.
\begin{lemma}
\label{lemma:nU_uncorrelated}
Assume that the DS algorithm is implemented indefinitely and $H_m$ is true.
Then, for each $\varphi \in \Theta^{(0)}$ and for every fixed $\epsilon>0$ there exist $C>0$ and $\gamma>0$ such that
\beq
\bea{l}
\label{eq:lemma:nU_uncorrelated}
\mathbf{P}_m\left(n_U(\varphi)>n\right)\leq C e^{-\gamma n} \vspace{0.1cm}\\\hspace{3cm}
\forall n>-(1+\epsilon)\log c/D\left(\theta^{(1)}\right)\;.
\ena
\eeq
\vspace{0.1cm}
\end{lemma}
\begin{proof}
Define $\widetilde{\ell}_m^{(\theta^{(1)},\varphi)}(t) = \ell_m^{(\theta^{(1)},\varphi)}(t) - D(\theta^{(1)}||\varphi)$. Using the same steps as in the proof of Lemma \ref{lemma:nU} (choosing this time $\epsilon_1 = D(\theta^{(1)}||\varphi)\epsilon/(1+\epsilon)>0$), equation (\ref{eq:Pm_nU}) holds with $\ell_m^{(\theta^{(1)},\varphi)}(t)$ and $\widetilde{\ell}_m^{(\theta^{(1)},\varphi)}(t)$ instead of $\ell_m(t)$ and $\widetilde{\ell}_m(t)$, respectively. Again, since $\widetilde{\ell}_m^{(\theta^{(1)},\varphi)}(t)$ has zero mean for all $t>\tau_{ML}$, all the three terms can be bounded as done in (\ref{eq:Pm_nU}). \\Since $D(\theta^{(1)}||\varphi) \geq D(\theta^{(1)}), \; \forall \varphi \in \Theta^{(0)}$, (\ref{eq:lemma:nU_uncorrelated}) follows. \\  
\end{proof} 
Using Lemma \ref{lemma:nU_uncorrelated} we have: \vspace{0.1cm} \\
$\mathbf{P}_m\left(n_U>n\right)\leq \mathbf{P}_m\left(n_U(\varphi)>n\right) \leq C e^{-\gamma n} \vspace{0.1cm}\\ 
\vspace{0.1cm} \hspace{4cm} \forall n>-(1+\epsilon)\log c/D\left(\theta^{(1)}\right),$ 
thus, the actual detection time is upper bounded by: \vspace{0.1cm}\\
$\displaystyle\mathbf{E}_m(\tau)\leq\mathbf{E}_m(\tau_{ML})+\mathbf{E}_m(n_U)\leq -\left(1+o(1)\right)\frac{\log(c)}{D(\theta^{(1)})}$, \\
and using the bound on the error probability obtained in step 1, the Bayes risk is upper bounded by:  
\begin{center}
$\displaystyle R_m(\Gamma)\leq -\left(1+o(1)\right)\frac{c\log(c)}{D(\theta^{(1)})}$. 
\end{center}
Combining the upper bound with the lower bound from \cite{Chernoff_1959_Sequential} completes the proof.
\QEDA
\subsection{Proof of Theorem~\ref{th:DS_generelized_policy}}
\label{app:proof_policy4}

For purposes of analysis we require that the stopping rule does not stop the test before $-\epsilon\log c$ samples have been taken, where $\epsilon>0$ is arbitrarily small. Also, when updating the sum MALLRs at time $t$, we use the current estimates for all $n=1, ..., t$.

Without loss of generality, let $H_m$ be the true hypothesis. Let $P_e=\sum_{m=1}^{M}\pi_m \alpha_m$ be the error probability, where
\begin{center}
$\alpha_{m,j}=\mathbf{P}_m(\delta=j)$
\end{center}
for all $j\neq m$.
Thus,
\begin{center}
$\displaystyle\alpha_m=\sum_{j\neq m}\alpha_{m,j}$.
\end{center}
Therefore, we need to show that $\alpha_{m,j}$ decays polynomially with $-\log c$. Note that $\alpha_{m,j}$ can be rewritten as follows:
\beq
\bea{l}
\displaystyle\alpha_{m,j}=\mathbf{P}_m\left(\delta=j\right) 
=\displaystyle\mathbf{P}_m\left(\delta=j, \tau_{ML}>\tau\right)\vspace{0.1cm} \\\hspace{3.5cm}
+\displaystyle\mathbf{P}_m\left(\delta=j, \tau_{ML}\leq\tau\right)\vspace{0.1cm}.
\ena
\eeq
Since the stopping rule does not stop the test before $-\epsilon\log c$ samples have been taken, the first term on the RHS is upper bounded by $C\tau^{-\gamma}\leq C(-\epsilon\log c)^{-\gamma}$, for some constants $C, \gamma, \epsilon>0$, resulting in a polynomial decay with $-\log c$. Thus, it remains to show that the second term on the RHS decreases polynomially with $-\log c$.

Accepting $H_j$ at time $n$ implies $S_j^{(0)}(n)+S_{m^{(1)}(n)}^{(1)}(n)\geq -\log c$, which implies $S_j^{(0)}(n)+S_m^{(1)}(n)\geq -\log c$. Hence, for all $j\neq m$ we obtain:
\beq
\bea{l}
\displaystyle\mathbf{P}_m\left(\delta=j, \tau_{ML}\leq\tau\right) \vspace{0.1cm} \\\hspace{0.0cm}
\leq\displaystyle\mathbf{P}_m\left(S_j^{(0)}(n)+S_m^{(1)}(n)\geq -\log c, \tau_{ML}\leq\tau\right)\vspace{0.1cm} \\\hspace{0.0cm}
\leq\displaystyle c\mathbf{P}_j\left(S_j^{(0)}(n)+S_m^{(1)}(n)\geq -\log c, \tau_{ML}\leq\tau\right)\leq c\vspace{0.1cm},
\ena
\eeq
where changing the measure in the second inequality follows by the fact that $S_j^{(0)}(n)+S_m^{(1)}(n)\geq -\log c$, and that the estimates are given by the true parameters for all $\tau\geq \max \{\tau_{ML}, \widetilde{\tau}_{ML}\}$ (where the current estimates are updated for all $n=1, ..., \tau$). As a result, the theorem follows.
\QEDA

\bibliographystyle{ieeetr}

\end{document}